\documentclass[runningheads]{llncs}
\usepackage[T1]{fontenc}

\newif\ifanonymous\anonymousfalse
\newif\ifdisablenotes\disablenotesfalse

\usepackage[a4paper,margin=1.3in]{geometry}
\usepackage{tikz,pgfplots}

\pgfplotsset{compat=1.15}
\usetikzlibrary{shapes,arrows,fit,calc,positioning,automata,decorations.pathreplacing,patterns,decorations.pathmorphing,decorations.markings}
\usetikzlibrary{cipher}

\usepackage{amsmath,amssymb,amsfonts}
\usepackage{enumerate,graphicx,url}
\usepackage[colorlinks=true, linkcolor=red, citecolor=blue,linktoc=section]{hyperref}
\usepackage{colortbl}
\usepackage{multirow}
\usepackage{verbatim}
\usepackage{booktabs}
\usepackage{xspace} 
\usepackage{lastpage}
\usepackage{tabu}
\usepackage{makecell}
\usepackage{comment}
\usepackage{csquotes}
\usepackage{mathtools}
\usepackage{xparse}
\usepackage{slashbox,pict2e}
\usepackage{authblk}
\usepackage{subcaption}
\ifdisablenotes%
\usepackage[disable]{todonotes}
\else%
\usepackage{todonotes}
\fi%
\usepackage[figuresright]{rotating}
\usepackage{pifont}
\makeatletter
\renewcommand*\env@matrix[1][*\c@MaxMatrixCols c]{%
	\hskip -\arraycolsep
	\let\@ifnextchar\new@ifnextchar
	\array{#1}}
\makeatother

\usepackage{threeparttable}
\usepackage{caption}
\usepackage{placeins}
\usepackage{siunitx}
\usepackage[capitalize,nameinlink]{cleveref}
\usepackage[skins]{tcolorbox}
\usepackage{orcidlink}

\newtcolorbox{sidebarbox}[2][]{
  enhanced,
  colback=white,            % white background
  colframe=#2,              % sidebar color
  boxrule=0pt,              % no border
  borderline west={3pt}{0pt}{#2}, % colored bar on the left
  sharp corners,
  left=15pt, right=6pt, top=6pt, bottom=6pt,
  fonttitle=\bfseries,
  title=#1
}

\newtcolorbox{myframe}[2][]{%
  enhanced,colback=cyan,colframe=black,coltitle=black,
  sharp corners,boxrule=0.7pt,
  fonttitle=\bf, attach boxed title to top left={yshift=-0.3\baselineskip-0.4pt,xshift=2mm},
  boxed title style={tile,size=minimal,left=0.5mm,right=0.5mm,
    colback=white,before upper=\strut},
  title=#2,#1
}

\newtcolorbox{simpleframe}[2][]{%
  enhanced,colback=white,colframe=black,coltitle=black,
  sharp corners,boxrule=0.7pt,
  fonttitle=\normalsize, attach boxed title to top left={yshift=-0.3\baselineskip-0.4pt,xshift=2mm},
  boxed title style={tile,size=minimal,left=0.5mm,right=0.5mm,
    colback=white,before upper=\strut},
  title=#2,#1
}

\Crefname{algocf}{Algorithm}{Algorithms}
\Crefname{appendix}{Appendix}{Appendix}

\DeclareMathOperator*{\concat}{%
    \mathchoice%
        {\Big\Vert}%
        {\big\Vert}%
        {\Vert}%
        {\Vert}%
}

\DeclarePairedDelimiterX{\rvect}[1]{[}{]}{\,\rowvec{#1}\,}
\ExplSyntaxOn
\NewDocumentCommand{\rowvec}{m}
 {
  \seq_set_split:Nnn \l_tmpa_seq { , } { #1 }
  \begin{pmatrix}
  \seq_use:Nn \l_tmpa_seq { & }
  \end{pmatrix}
 }
\ExplSyntaxOff

\tikzset{XOR/.style={draw,circle,scale=1,append after command={
        [shorten >=\pgflinewidth, shorten <=\pgflinewidth,]
        (\tikzlastnode.north) edge (\tikzlastnode.south)
        (\tikzlastnode.east) edge (\tikzlastnode.west)
        }
    }
}
\tikzset{MUL/.style={draw,circle,scale=1,append after command={
        [shorten >=\pgflinewidth, shorten <=\pgflinewidth,]
        (\tikzlastnode.north west) edge (\tikzlastnode.south east)
        (\tikzlastnode.south west) edge (\tikzlastnode.north east)
        }
    }
}
\tikzset{buswidth/.style={draw=none,circle,scale=1,append after command={
        [shorten >=\pgflinewidth, shorten <=\pgflinewidth,]
        (\tikzlastnode.south west) edge (\tikzlastnode.north east)
        }
    }
}
\tikzset{enc/.style=   {draw=black,fill opacity=.3},
         encdec/.style={draw=black,fill opacity=1}
}
\definecolor{darkorange}{HTML}{FF8C00}
\definecolor{navyblue}{HTML}{000080}
\definecolor{britishracinggreen}{HTML}{004225}
\newcommand{\prob}{\text{Pr}}

\newcommand{\F}{\mathbb{F}}
\newcommand{\Z}{\mathbb{Z}}
\newcommand{\beneswritten}{Bene\v s}

\newcommand{\floor}[1]{\left\lfloor #1 \right\rfloor} % custom floor with both brackets
\newcommand{\ceil}[1]{\left\lceil #1 \right\rceil} % custom ceil with both brackets
\newcommand{\abs}[1]{\lvert #1 \rvert} % custom absolute value symbol

\newcommand{\lm}{\mathsf{LM}}

\newcommand{\bigO}{\mathcal{O}}
\renewcommand{\deg}{\ensuremath{\mathsf{deg}}}

\newcommand{\cico}{\mathsf{CICO}}

\newcommand{\vars}[1]{\mathbf{#1}}

\newcommand{\benes}{\ensuremath{\mathcal{B}}}

\newcommand{\roundfunction}{\mathcal{R}}
\newcommand{\nllayer}{\mathcal{F}}
\newcommand{\llayer}{\mathcal{M}}

\newcommand{\trunc}{\texttt{Trunc}}

\newcommand{\design}{\text{BenX}}

\newcommand{\monolith}{\texttt{Monolith}}

\newcommand{\poseidon}{\textsc{Poseidon}}
\newcommand{\poseidontwo}{\textsc{Poseidon}2}

\newcommand{\rescueprime}{\textit{Rescue-Prime}}
\newcommand{\rpo}{\textit{Rescue-Prime Optimized}}

\begin{document}
\title{\design{}: Resource-Sharing Permutations for Computational Integrity}

\author{Luca Campa\inst{1}\orcidlink{0009-0004-6916-1918} \and
Thomas De Cnudde\inst{2}\orcidlink{0000-0002-2711-8645} \and
Al Kindi\inst{3} \and
Arnab Roy\inst{4}\orcidlink{0000-0002-3284-7076} \and
Fabian Schmid\inst{5}\orcidlink{0009-0005-7783-0859} \and
Markus Schofnegger\inst{6}\orcidlink{0009-0003-9598-4005} \and
Stefano Trevisani\inst{7}\orcidlink{0000-0002-5375-0214}}

\authorrunning{\phantom{L. Campa et al.}}

\institute{\email{campa1102@gmail.com} \and
Fhenix, Tel Aviv, Israel\\\email{thomas.decnudde@proton.me} \and
Miden, Dubai, United Arab Emirates\\\email{al.kindi@miden.team} \and
\email{arnab.roy@windowslive.com} \and
Technical University of Graz, Graz, Austria\\\email{fabian.schmid@tugraz.at} \and
[[alloc] init], New York, USA\\\email{markus.schofnegger@gmail.com} \and
\email{stefanotrevisani98@gmail.com}}
\maketitle              % typeset the header of the contribution

\begin{abstract}
Cryptographic hash functions over integers modulo a prime play a decisive role in the efficiency and security of proof systems for computational integrity. Early designs focused on compact arithmetic circuits and efficient software execution, primarily targeting general-purpose CPUs rather than hardware accelerators. This work focuses on enabling efficient resource sharing and hardware acceleration alongside efficient software execution.
  
We propose \design{}, a permutation-based hash function designed for hardware acceleration, fast CPU execution, and low circuit complexity in proof systems. To construct the underlying permutation, we turn the Bene\v{s} network into an invertible function over $\F_p^2$ using a Dickson polynomial. Its structure yields a permutation particularly suited to hardware resource sharing and acceleration.

Fully pipelined on FPGA, \design{} matches the throughput of \poseidontwo{}, with $1.36\times$ and $2.15\times$ lower latency for Goldilocks and BabyBear, respectively. Although larger as a standalone core, it makes more effective use of shared hardware: in our dual-mode architecture, where the NTT and hash share field multipliers, \design{} keeps 100\% of them busy, compared with 6.25--20.3\% for the partial rounds of \poseidon{} and \poseidontwo{}.

In software, \design{} outperforms \poseidon{} in all tested 31- and 64-bit configurations, but is slower than \poseidontwo{}, except for the 16-element BabyBear instance. In zero-knowledge proofs, \design{} proves Goldilocks permutations $6$--$10\times$ faster than \monolith{}. Its compact arithmetization uses about $36\%$ fewer trace cells than \poseidon{} and \poseidontwo{} over BabyBear, while its fast variant requires $1.3$--$2.6\times$ the prover time of \poseidontwo{}.
\end{abstract}

\section{Introduction}
\label{section:introduction}
Since their conception~\cite{DBLP:conf/asiacrypt/Albrecht0R0T16}, Arithmetization-Oriented (AO) cryptographic primitives over prime fields have become an important building block of modern zero-knowledge (ZK) proof systems. Due to compact algebraic structure, these constructions can substantially reduce the cost steming from
the underlying constraint system.
Recently, a number of ZK proof systems has focused on `small' prime fields whose arithmetic enables efficient use of the native datapaths and computational resources of modern hardware. 
Prominent examples include Plonky2~\cite{PolygonZero2022Plonky2}, which employs the 64-bit `Goldilocks' field, and the RISC Zero zkVM~\cite{BruestleGafniRiscZero2023}, which instantiates its STARK over the 31-bit `BabyBear' field. More recently, Circle STARKs~\cite{HabockLevitPapini2024Circle} enabled efficient STARK constructions over the Mersenne prime $p=2^{31}-1$.
The growing adoption of small-field proof systems in turn motivated the corresponding constructions of AO primitives. 
For example, the Poseidon2~\cite{GrassiKhovratovichSchofnegger2023Poseidon2} hash function provides dedicated Goldilocks and BabyBear instances, while Tip5~\cite{SzepieniecEtAl2023Tip5} targets Goldilocks specifically, illustrating the efforts toward small-field-aware design.\\

\noindent
\textsc{Hardware Utilization in ZK Proof Systems.} Despite significant progress at the design level, the scaling of ZK proof systems continue to impose substantial computational and architectural demands. This motivates a holistic approach to algorithm-hardware co-design across the cryptographic stack.

The importance of hardware efficiency is increasingly evident in contemporary ZK implementations, where prover performance is often limited by a small set of highly repetitive arithmetic kernels. In such a context, FPGA-based implementations become particularly attractive: beyond parallelism opportunities, they enable fine-grained specialization of finite-field arithmetic, configurable datapaths, and explicit sharing of expensive resources such as modular multipliers, reduction units, memories, and interconnect. Such resource sharing is especially important in ZK workloads, where area efficiency and memory bandwidth frequently constrain the degree of exploitable parallelism. 
Indeed, substantial effort \cite{ZPrize2022,ZKsyncGPU,RISCZeroCUDA,SP1GPU} has been devoted to explore FPGA-based (and GPU-based) acceleration of dominant ZK frameworks components.

However, when considering existing AO primitives, their internal structure may lead to inefficient hardware utilization when implemented directly, for example by requiring many concurrent nonlinear units. 
This observation is explained by the fact that, while substantial effort has been devoted to accelerating the surrounding proof-system machinery~\cite{KocerKTAS2025,Cammarota2022,DBLP:conf/fpga/AasaraaiCMSVB23,DBLP:conf/isca/ZhangWZDMLWZG021}, comparatively less attention has been paid to the design of AO primitives themselves with efficient hardware realization as a first-class objective. This gap has only recently been recognized in RainHash2.0~\cite{DBLP:journals/tches/CuiGKKMRRSS26} by addressing hardware friendliness of AO hash functions by design.\\

\noindent
\textsc{Resource Sharing with NTT.}
One of the most central computational kernels for ZK proof systems is the Number Theoretic Transform (NTT), which is used extensively for polynomial evaluation, interpolation, and low-degree extensions (LDE). Consequently, NTT computation represents a major target for ZK hardware acceleration~\cite{DBLP:conf/isca/ZhangWZDMLWZG021,ZhaoEtAl2023NTT}, as can be seen from practical FPGA-oriented efforts like ZPrize~\cite{ZPrizeNTT}, and existing ZK accelerators devote substantial hardware resources to NTT execution, often alongside other dominant prover kernels such as Multi-Scalar Multiplication (MSM) and hashing.

A pipelined NTT datapath consists mainly of butterfly units, routing logic, and storage, with modular multiplication among its most expensive elements. High-throughput designs therefore instantiate multiple multipliers and aim to keep them highly utilized~\cite{ZhaoEtAl2023NTT,DBLP:conf/isca/ZhangWZDMLWZG021}. This motivates reusing these arithmetic units across prover phases: similar resource sharing between NTT and MSM has already been explored~\cite{ChenEtAl2025NTTMSM}. We investigate whether the same principle can be extended to hashing, utilising the NTT datapath for the hash function with minimal additional hardware.

Thus, the foundational architectural question behind our work is: \textit{Can an AO hash function be designed to map efficiently onto the computational structure of an NTT?} This objective imposes several constraints on the hash design. First, state elements should undergo a similar number of multiplications, allowing the available multipliers to be utilized uniformly, thus making highly asymmetric constructions, like extensive partial rounds, less attractive. Second, bitwise operations and lookup tables should be avoided, as they require functionality that is not naturally provided by an NTT datapath. Finally, nonlinear layers requiring many sequential multiplications, such as large or inverse power maps, are undesirable, since they either require additional arithmetic resources or multiple passes through the shared datapath. These considerations motivate an AO primitive whose computational structure is closely aligned with that of the NTT.

\paragraph{Related Work.} Hash functions related to the NTT have been proposed in several works, albeit with different targets in mind. 
SWIFFT in particular stands out as a compression function with no practical attacks known to date \cite{lyubashevsky2008swifft}. It evaluates a structured linear map (a fixed negacyclic convolution in \(R_q = \mathbb{Z}_q[x]/(x^n+1)\)) realizable efficiently with an NTT, with collision resistance based on the hardness of ideal-SIS \cite{lyubashevsky2008swifft}. However, its inputs are binary vectors, reducing efficiency per hashed bit, and its linear structure means it lacks pseudorandomness and is not typically modeled as a random oracle, limiting its suitability for ZK applications. Vision Mark-32~\cite{DBLP:journals/iacr/AshurMPS24} is another hardware-friendly design and was designed with binary fields in mind, whereas we look specifically into prime field hash functions. 

\subsection{Our Results}
In this work, we introduce \design{}, a new AO permutation over $\F_p$ for small primes $p$, whose structure is based on the \emph{\beneswritten{} network}.
Through efficient hardware resource sharing, our design aims to achieve competitive performance across proof-system arithmetizations, native software, and dedicated hardware, hence being a more balanced solution in the cryptographic stack.
Along the way of pursuing our task, we put forward several contributions.\\

\noindent
\underline{\textsc{An Invertible \beneswritten{} Construction.}}
To construct our permutation-based hash, we start from the Bene\v{s} network~\cite{benesnetwork}
\[
\benes(x,y)=\bigl(f_0(x)+f_2(y),f_1(x)+f_3(y)\bigr),
\]
where the $f_i$ are arbitrary functions. Since this mapping is not invertible in general, we impose the conditions $f_2=f_3$ and require $f_0-f_1$ to be a permutation polynomial over $\F_p$ (see \cref{subsection:non-linear-building-blocks}). Under these conditions, subtracting the two output coordinates eliminates the (y)-dependent term and yields $((f_0-f_1)(x))$, from which $x$, and subsequently $y$, can be recovered uniquely. This turns the Bene\v{s} structure into a permutation over $\F_p^2$. In contrast to Feistel~\cite{DBLP:journals/siamcomp/LubyR88} and Lai--Massey~\cite{DBLP:conf/eurocrypt/LaiM90,DBLP:conf/asiacrypt/Vaudenay99} constructions, whose invertibility relies on preserving a branch, the \beneswritten{} construction transforms both branches at every application.\\

\noindent
\underline{\textsc{Efficient Instantiation with Dickson Polynomials.}}
For an efficient instantiation of $f_0, f_1$ while satisfying the condition that \(P = f_0 - f_1\) is a PP, we utilise the Dickson polynomial of first kind. By choosing
$P = D_d(x,1)$, the Dickson polynomial~\cite{LidlNiederreiter1997FiniteFields,LidlMullenTurnwald1993} of first kind having degree $d$, we are able to split $P$ into $f_0 = x^d$, $f_1 = x^d - D_d(x,1)$. We also impose the necessary conditions on $d$ to ensure \(D_d(x,1)\) and $x^d$ are PPs over $\F_p$. 
By Schur's conjecture (proven in~\cite{Fried1970Schur,Turnwald1995Schur}), an integer polynomial that permutes \(\F_p\) for infinitely many primes \(p\) is
a composition of linear polynomials, power maps and Dickson polynomials. Thus, we choose the two simplest nonlinear functions (over $\F_p$) for instantiating the Bene\v{s} network in our construction. Finally, the value of $d$ is chosen to ensure efficient hardware support. 

We finally stress that the well-known linearization of Dickson polynomials, namely \(D_d(z + z^{-1}, 1) = z^d + z^{-d}\) over \(\F_{p^2}\), which explains the condition \(\gcd(d, p^2-1) = 1\), and which earlier Dickson-based public-key proposals relied upon~\cite{MullerNobauer1985}, is not inherited by our round function. 
Indeed, \(D_d\) is never applied to the state; it only arises as the difference of the two output words of a single block, so that neither \(x^d + y^d\) nor \(f_1(x) + y^d\), is a Dickson polynomial.\\

\noindent
\underline{\textsc{Linear Transformation over $\F_{p^2}$.}}
Since each \beneswritten{} block defines a mapping $\F_p^2 \mapsto \F_p^2$, we treat each pair as a single element of the quadratic extension \(\F_{p^2}\) for the application of the linear transformation. This linear transformation is defined by an MDS matrix over \(\F_{p^2}\) (with entries in \(\F_p\)), and applied after each nonlinear transformation. This approach halves the cost of a full MDS application over the entire state. Combined with the diffusion already provided by the butterflies, this mixes all branches together, while the MDS property facilitates statistical analysis and helps against recent round-skipping attacks exploiting non-MDS linear layers~\cite{DBLP:journals/iacr/MerzG26}.\\

\noindent
\underline{\textsc{Cryptanalytic Security.}}
We support the security of \design{} with statistical and algebraic cryptanalysis.
In particular, we provide an upper bound on the algebraic cryptanalysis complexity based on the $\F_q$-dimension of the quotient ring. We note that the security (parameter) of \design{} is not inferred from the complexity of basis computation. The parameter choices for \design{} e.g. number of rounds, number of output branches ($t$) consisting the compression (or hash) function output, over specific $\F_p$, are derived from both statistical and algebraic cryptanalysis. In particular $t$ is chosen to ensure security against algebraic cryptanalysis for small (e.g. 32 and 64 -bit) primes $p$.\\

\noindent
\underline{\textsc{From NTT Butterflies to a \beneswritten{}-Based Round Function.}}
We exploit the connection between the NTT computation and circuit-friendly hashing in hash-based proof systems. Since hashing follows the LDE step in the computation flow, we design a symmetric permutation that reuses the NTT building blocks, and particularly its multipliers, via \beneswritten{}-based nonlinear layers. The \beneswritten{} construction generalizes the NTT butterfly by allowing arbitrary functions on each of the two branches, letting us embed a nonlinear layer into the same pairwise dataflow already used for the LDE. The resulting design has a simple algorithmic description with identical rounds, each following the well-known nonlinear/liner layer alternation pattern. 
Hardware implementations that target resource-sharing with the NTT additionally benefit from this design because of two factors: balanced computations across all state elements, and omission of bitwise operations and inversions. This allows to maximally reuse the NTT's multipliers with the least resource overhead. 

We stress that the NTT itself is not used in the permutation. The key point is that both the NTT and \beneswritten{} permutation share a common butterfly structure, which we leverage to maximize resource sharing by time-multiplexing with the NTT engine, hence minimizing area occupancy.\\
 
\noindent
\underline{\textsc{Evaluation in Hardware, Software, Proof System.}}
We evaluate \design{} in all three settings it targets (\cref{section:benchmarks}). In hardware, we implement all four instances as fully unrolled, fully pipelined FPGA cores and compare the Goldilocks and BabyBear instances with \poseidon{} and \poseidontwo{}. Over Goldilocks, \design{} runs at 478.5 MHz, delivering 3.1\% and 5.0\% higher throughput than \poseidon{} and \poseidontwo{}, respectively, together with 22.1\% and 26.7\% lower latency, at the cost of a larger standalone core. Over BabyBear, \design{} essentially matches the throughput of \poseidon{} with 32.9\% fewer LUTs and 26.6\% fewer FFs, while cutting the cycle count by 41.3\% and the latency by 41.1\%, a 1.7x improvement. Against \poseidontwo{}, whose throughput is only 4.2\% higher, a modest increase in area buys 55.4\% fewer cycles and 53.6\% lower latency, a 2.15x improvement.

The advantage of \design{} stands out in our dual-mode NTT/hash architecture, in which the NTT and the hash tile share a single pool of field multipliers. Relative to separate NTT and hash tiles, sharing the pool saves 42.9\% of the DSPs, often the limiting resource in FPGA designs for ZK (and FHE), for all shared variants. An exclusive feature of \design{} is its full utilization of the resource pool: its uniform main-round tile keeps every multiplier busy, i.e., 100\% multiplier saturation, whereas the partial rounds of \poseidon{} and \poseidontwo{} occupy only 6.25--20.3\% of it. Thus, \design{} turns the arithmetic an NTT accelerator already contains into a consistently utilized hash datapath, instead of leaving most multipliers idle during partial rounds.

In software, \design{} offers competitive performance across several prime fields, compared to various hash functions. Our Plonky3 benchmarks show speedups of approximately $7$--$13\times$ over \rescueprime{} and up to $2\times$ over \poseidon{}, which it outperforms in all available instances. Compared to the split-and-lookup design \monolith{}, \design{} performs worse for narrow state sizes, but approximately $1.9\times$ faster for the 24-element Mersenne31 state, while requiring exclusively algebraic operations. Against \poseidontwo{}, \design{} achieves a modest $4\%$ runtime reduction for BabyBear at $t=16$, but incurs $12$--$66\%$ higher runtimes elsewhere. These costs follow our choice of uniform full rounds and an MDS linear layer over paired state elements, motivated by hardware resource sharing and structural resilience against round-skipping attacks~\cite{DBLP:journals/iacr/MerzG26}. An additional comparison with a broader range of Goldilocks hash functions appears in \cref{appendix:performance}.

The remainder of the paper is organized as follows. \cref{section:background} covers preliminaries and background, \cref{section:designrationale} and \cref{section:design} present the design rationale and the permutation, \cref{section:securityanalysis} analyzes security, and \cref{section:benchmarks} evaluates the practical performance.

\section{Hardware Acceleration and NTT Resource Sharing}
\label{section:background}
Hardware implementations typically decompose a protocol into performance- or power-critical primitives, which may then be mapped to dedicated accelerators. Hardwired datapath accelerators specialize the datapath for a fixed kernel or small set of kernels, with only limited runtime configurability.
In particular, accelerators proposed in the ZK space are typically hardwired datapath accelerators, featuring multiple kernels including the NTT (\cref{fig:polymul}), Merkle Tree hashing, and Multi-Scalar Multiplication.
\begin{figure}[h]
    \centering
    \includegraphics[width=0.9\textwidth]{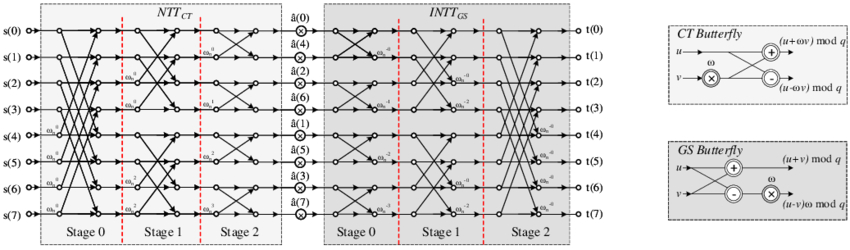}
    \caption{A polynomial multiplication using an NTT, a pointwise multiplication, and an inverse NTT~\cite{bisheh2021high}.}
    \label{fig:polymul}
\end{figure}

Importantly, during the LDE in hash-based STARKs, the inverse NTT interpolates evaluation points to polynomial coefficients, after which the polynomial is evaluated over a larger domain by NTT extrapolation.
Then, AO hashing finalizes the commitment after the LDE step and makes the opening more efficient in a recursive prover.
When several primitives of a protocol occupy distinct parts of the silicon, further optimization is possible through \textit{resource sharing}, where primitives share the same hardware resources wherever the algorithm allows it, reducing total silicon area and thus cost, and increasing adoption.\footnote{In the asymptote, this describes a CPU, where a handful of compute units are shared across any computable algorithm.} 
In our case, where hashing strictly starts after the NTTs are computed, sharing resources between these two algorithms can significantly reduce silicon area, and the saved hardware could be filled with more resources for higher performance, see~\cref{fig:chip-area}.

\begin{figure}[htbp]
    \centering
    \includegraphics[width=0.8\textwidth]{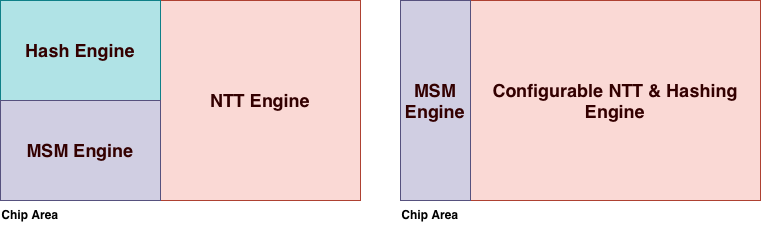}
    \caption{Sharing resources between the NTT and an AO hash function can lead to using more resources towards both algorithms.}
    \label{fig:chip-area}
  \end{figure}

This motivates the design of a AO hash function that can `maximally reuse' the components of an NTT, while minimizing the introduction of additional computational elements.
Overall, an NTT requires, as fundamental building blocks, \emph{butterfly units} 
    consisting of one modular addition, one modular subtraction, and one modular multiplication,
    \emph{multiplexers} to route coefficients to form a spatial permutation, and \emph{delay elements} to route coefficients in time to form a temporal permutation. 
Of these, the modular multiplications in the butterfly units are the most resource-intensive, and sharing this resource is a guiding principle of the design of our AO hash function. Mapping the hash function efficiently on a pipelined NTT imposes the natural constraints:
\begin{enumerate}
  \item Each state element should go through a similar number of multipliers to allow full reuse of the NTT's multipliers. Partial rounds are therefore discouraged.
  \item Bitwise operations or Lookup Tables should be avoided as they require logic that incurs overheads on the NTT. 
  \item High-degree power maps should be avoided. They are very expensive to unroll fully in hardware, and when not unrolled, state elements would need to go through multiple passes to be computed, increasing control overheads. 
\end{enumerate}

\section{Design Rationale}
\label{section:designrationale}
In this section we provide the mathematical details of the different functions used for constructing the \design{} permutation, that in turn is used to construct the hash function. 

\subsection{From an Invertible Bene\v{s} Network to the Nonlinear Transformation}
\label{subsection:non-linear-building-blocks}
The \beneswritten{} network, defining a mapping $\F_p \times \F_p \mapsto \F_p \times \F_p$, is non-invertible in general (depicted in \cref{fig:benes-base}).

\begin{definition}[\beneswritten{} Network~\cite{benesnetwork}]\label{defn:benes-construction}
    Let \(\F_p\) be a prime field. Let \(f_0, f_1, f_2, f_3 : \F_p \to \F_p\) be 4 generic functions. Let \(x,y \in \F_p\). The \beneswritten{} construction is the function \(\benes_{f_0, f_1, f_2, f_3} : \F_p^2 \to \F_p^2\) defined by the map:
    \[
        (x,y) \mapsto (f_0(x) + f_2(y), f_1(x) + f_3(y)).
    \]
\end{definition}
Note that we will drop the subscript from $\mathcal{B}$ for simplifying the notation. For making $\mathcal{B}$ invertible we must choose \(f_0, f_1, f_2, f_3\) carefully.

\begin{figure}
    \centering
    \begin{tikzpicture}[scale=0.7]

%% Horizontal 

\node at (0, 0) (inputtop) {$X_0$};
\node at (0, -3) (inputbot) {$X_1$};

\node at (2, 0) [draw,minimum width=0.5cm,minimum height=0.5cm] (F1) {$f_0$};
\node at (1, -1) [draw,minimum width=0.5cm,minimum height=0.5cm] (F2) {$f_1$};
\node at (1, -2) [draw,minimum width=0.5cm,minimum height=0.5cm] (F3) {$f_2$};
\node at (2, -3) [draw,minimum width=0.5cm,minimum height=0.5cm] (F4) {$f_3$};

\node at (4, 0) [inner sep=0pt] (addtop) {\Large$\boxplus$};
\node at (4, -3) [inner sep=0pt] (addbot) {\Large$\boxplus$};

\node at (5, 0) (outputtop) {$Y_0$};
\node at (5, -3) (outputbot) {$Y_1$};

\draw[->] (inputtop) -- (F1);
\draw[->] (F1) -- (addtop);
\draw[->] (inputbot) -- (F4);
\draw[->] (F4) -- (addbot);
\draw[->] (inputtop) -- ($(inputtop)+(1,0)$) -- (F2);
\draw[->] (F2) -- ($(F2)+(0.5,0)$) -- ($(addbot)+(0,+0.5)$) -- (addbot);
\draw[->] (inputbot) -- ($(inputbot)+(1,0)$) -- (F3);
\draw[->] (F3) -- ($(F3)+(0.5,0)$) -- ($(addtop)+(0,-0.5)$) -- (addtop);
\draw[->] (addtop) -- (outputtop);
\draw[->] (addbot) -- (outputbot);

\end{tikzpicture}
    \caption{A \beneswritten{} network.}\label{fig:benes-base}
\end{figure}
 
\begin{proposition}\label{prop:benes-invert}
    Let $x,y,u,v \in \F_p$ and let \(\benes: \F_p^2 \to \F_p^2\) with \(f_0, f_1, f_2, f_3:  \F_p \to \F_p\) be such that \(f_2 = f_3\), and \(f_2\) and \(f_0 - f_1\) are PPs. Then, \(\benes\) is invertible, meaning that \( \benes(x,y) = (u,v) \iff \benes^{-1}(u,v) = (x,y)\) and \(\benes\) is an injective function in \(\F_p^2\).
\end{proposition}

\begin{proof}
    Let \(x, y \in \F_p\) be two generic values. \(u,v = \benes(x, y)\). By definition, \(u = f_0(x) + f_2(y)\) and \(v = f_1(x) + f_3(y)\). Let us define \(\delta = u - v = f_0(x) - f_1(x)\). Note that \(f_2(y) = f_3(y)\) cancel each other out. By definition, \(f_0 - f_1\) is a PP in \(\F_p\), therefore there is one and only one value \(\tilde{x}\) such that \(f_0(\tilde{x}) - f_1(\tilde{x}) = \delta\). Then, this value must be the original input \(x\).
    Once \(x\) is found, without loss of generality we can take into consideration either \(u\) or \(v\). Let us take \(u\) into account, then \(u = f_0(x) + f_2(y)\), but \(m = f_0(x)\) is now known. \(f_2(y) = u - m\). By definition, \(f_2\) is a PP in \(\F_p\), therefore there is one and only one value \(\tilde{y}\) such that \(f_2(\tilde{y}) = u-m\). This value must be the original input \(y\).
\end{proof}

\noindent
\textbf{Choosing the \boldmath$f_i$.}
While \Cref{prop:benes-invert} guides the choice of \(f_0, f_1, f_2, f_3\), the choice of the PPs still depends on the finite field or the prime $p$ and must provide minimal number of field multiplications for efficiency. In addition, the choice of
permutation polynomial also has effect on the security analysis. Given a prime $p$, many classes of PPs exist~\cite{LidlNiederreiter1997FiniteFields}.  However, to avoid the complexity of choosing (different) PPs for each given prime, we want a single algebraic description of the nonlinear transformation that instantiates over a family of primes (\cref{section:design}). 
For solving this issue we turn to the Schur's conjecture that is proven by Fried~\cite{Fried1970Schur} and also by Turnwald~\cite{Turnwald1995Schur}. 
Schur's theorem states that a polynomial \(f \in \Z[x]\) permuting \(\F_p\) for infinitely many primes \(p\) is necessarily a composition of linear polynomials, power maps \(x^d\), and Dickson polynomials \(D_d(x,a)\) with \(a \in \Z\). We stress that this statement does not say that no other polynomial permutes any one of our target fields, only that a prime-independent description of the nonlinear layer cannot escape these two families. 

We will choose the two basic nonlinear permutations (from Schur's theorem) for instantiating $f_i$ -- the power mapping $x \mapsto x^d$ and the Dickson's polynomial.\\

\noindent
\textit{Power Functions.}
For \(p\) prime and \(d \in \mathbb{N}\), \(x^d\) is a PP in \(\F_p\) if and only if \(\gcd(p-1, d) = 1\).\\

\noindent
\emph{Dickson's Permutation Polynomials.}
Dickson's polynomials of the first kind, defined as follows, yield PPs under certain condition.
\begin{definition}[Dickson Polynomial~\cite{Dickson1896}]
    Given \(a \in \F_p\), the Dickson polynomial of degree \(d\) is given by
    \begin{equation*}
        D_d(x,a) = \sum_{j=0}^{\lfloor d/2 \rfloor} \frac{d}{d-j}\binom{d-j}{j}(-a)^jx^{d-2j}.
    \end{equation*}
\end{definition}
If \(a \neq 0\) and \(d > 1\), then \(D_d(x, a)\) is a PP over $\F_p$ $\iff$ \(\gcd(d, p^2-1) = 1\)~\cite{shallue2012permutationpolynomialsfinitefields}.

\begin{definition}[Invertible Bene\v{s}]
    Let \(p\) be an odd prime and let \(d_1, d_2 \in \Z_{\phi(p)}\) be two integers such that \(D_{d_1}(x,1), x^{d_2}\) are PPs in \(\F_p\). Moreover, let \(D_{d_1}(x,1) = x^{d_1} + q(x)\) where \(q(x)\) is the remaining part of Dickson's polynomial of degree \(d_1\). The \beneswritten{} functions can be chosen as
    \[
        f_0(x) = x^{d_1}, \quad
        f_1(x) = -q(x), \quad
        f_2(x) = f_3(x) = x^{d_2}.
    \]
\end{definition}

\begin{example}
    Let \(p = 37\). The minimum \(d_1\) for which \(\gcd(d_1, p^2-1) = 1\) is 5. Then, let \(d_1 = d_2 = 5\), \(D_5(x,1) = x^5 -5x^3 + 5x\) then \(q(x) = -5x^3 + 5x\).
    The \beneswritten{} functions are defined as follows:
    \[
        f_0(x) = x^5, \quad
        f_1(x) = -(-5x^3 + 5x), \quad
        f_2(x) = f_3(x) = x^5.
    \]
  \end{example}

The \textit{nonlinear transformation} $\mathcal{F}$ is defined by concatenating say $\ell$, $\mathcal{B}$ functions. This means $\mathcal{F}$ can only process inputs (viewed as field elements) with even length ($2\ell$).   

\paragraph{Statistical Properties.}
 Differential cryptanalysis~\cite{DBLP:conf/crypto/BihamS90} utilises the distribution of output differences corresponding to a specific input difference for a function. 
Let $\alpha, \beta \in \F_p$, the differential probability of obtaining \(\beta\) from \(\alpha\) is defined as
\[
\text{Prob}(\alpha \rightarrow \beta)
=\max_{\substack{\alpha, \beta \neq 0}} \frac{\left|\{ x \in \mathbb{F}_p \mid f(x + \alpha) - f(x) = \beta\} \right|}{\abs{\mathbb F_p}}.
\]
The differential probability of the $f_i$ functions must be considered towards the security analysis of \design{}. Since these functions are polynomials the following lemma is used to derive an upper bound on the differential probability for each $f_i$.
\begin{lemma}
\label{lemma:diffprobsbox}
The differential probability of any function of degree \(d\) is upper-bounded by \(d - 1\) over the domain of the function, i.e., \(\frac{d - 1}{\abs{\mathbb F_p}}\).
\end{lemma}
The proof is immediate. For \(\deg(f) = d\), the difference polynomial \(f(x) - f(x + \alpha)\) with \(\alpha \in \mathbb F\) has degree \(d-1\) and hence at most \(d-1\) roots in the field. Thus, the Dickson polynomials we use have a maximum differential probability of \(\leq \frac{d - 1}{\abs{\mathbb F_p}}\).

\subsection{Linear Transformation}
The linear transformation ($M$) in the \design{} permutation will follow the nonlinear transformation ($\mathcal{F}$) that is derived from Bene\v{s} network i.e. $M$ is applied to the output of $\mathcal{F}$. Instead of defining $M$ over $\F_p$, we define it over $\F_{p^2}$.
The \(2l\) field elements that are outputs of $\mathcal{F}$ are grouped into \(l\) pairs which are then interpreted as elements of \(\F_{p^2}\). After fixing a basis \(\{1,\theta\}\) of \(\F_{p^2}\) over \(\F_p\), the pair \((u,v) \in \F_p^2\) represents \(u\theta + v \in \F_{p^2}\). We then apply the \(l \times l\) MDS matrix $M$ to these extension-field words.

Consider the \beneswritten{} outputs \(\{(u_i, v_i)\}_{i=0}^{l-1}\) with \(u_i, v_i \in \F_p\). Let \(M \in \F_p^{l \times l}\) be an MDS over \(\F_p\), viewed inside \(\F_{p^2}\) via the natural embedding \(\F_p \hookrightarrow \F_{p^2}\). Every square minor of \(M\) is an element of \(\F_p\), and every nonzero element of \(\F_p\) remains nonzero in \(\F_{p^2}\). Hence, the same matrix is MDS over \(\F_{p^2}\). Because all entries of \(M\) lie in the base field, applying \(M\) to \((u_i\theta + v_i)_i\) is equivalent to applying \(M\) separately to \((u_i)_i\) and \((v_i)_i\) and then recombining the two coefficient vectors.

The structure of the \beneswritten{} network already combines a pair of inputs. Even when \(M\) is applied separately to the two output sets, all \beneswritten{} inputs are therefore mixed together, providing full diffusion. For the initial affine layer, since only full diffusion is needed, we group the inputs in pairs, apply a simple Cooley--Tukey network to each pair, and then apply \(M\). This achieves full diffusion with a very simple and efficient linear layer. From a security perspective, our decision of using an MDS matrix (instead of simple branch permutation or non-MDS matrix) is motivated by the recent cryptanalysis exploiting non-MDS linear transformation~\cite{DBLP:journals/iacr/MerzG26}. It further simplifies the statistical analysis.

\paragraph{Selecting the MDS Matrices.}
We define \(M\) as a circulant matrix over \(\F_p\), that is determined by one vector generator \(c \in \Z^n\), reduced modulo \(p\). This gives a compact specification and admits structured implementations of the matrix--vector product. The implementation uses dimensions \(n \in \{4,6,8,12\}\) depending on the state width, and the reproducible selection procedure below applies to all four dimensions.

All rows listed here use the implementation's first-row convention. The supplementary material uses the transpose first-column convention for the polynomial derivation, so converting between the table and the supplementary material reverses entries \(1,\ldots,n-1\) of the generator. To make the choice of \(c\) reproducible, we fix the following public procedure and run it separately for each dimension.
\begin{enumerate}
    \item \textbf{Generate candidates.} For \(n \in \{4,6\}\), enumerate all \(c \in (\{-B,\dots,B\}\setminus\{0\})^n\) with \(B=4\). For \(n \in \{8,12\}\), the analogous row scans are too large, so we instead sample \(1.5 \cdot 10^{7}\) candidates in the FFT-block parameterisation used by the implementation; this parameterisation is described in \cref{appendix:mds-implementation}. The two sampled searches are deterministic and use SHA-256-derived RNG seeds from \texttt{"BenX-MDS-NUMS-N8-v1"} and \texttt{"BenX-MDS-NUMS-N12-v1"}.
    \item \textbf{Filter to MDS.} Keep only candidates that satisfy the MDS test described below. The selected rows are verified over all target fields.
    \item \textbf{Reduce by symmetry.} For the exhaustive \(n \in \{4,6\}\) searches, quotient by the adder-cost-preserving symmetries of circulant generators: cyclic rotation, reflection, global negation, and global power-of-2 scaling. These transformations preserve MDS-ness over the target odd-prime fields and preserve the adder count in our free-shift/free-negation cost model, up to fixed input/output permutations and global output scaling. We keep the lexicographically smallest representative of each orbit. For the sampled \(n \in \{8,12\}\) searches, we do not apply orbit canonicalisation; determinism comes from the fixed seed and the ranking rule below.
    \item \textbf{Rank candidates.} Within each dimension, sort lexicographically by
    \[
        \bigl(\text{cost}_n(c),\, \lVert c \rVert_1,\, \max_i \lvert c_i \rvert,\, c\bigr).
    \]
    The cost \(\text{cost}_n(c)\) is the direct circulant shift-and-add score for \(n \in \{4,6\}\), and the full FFT-block pipeline score for \(n \in \{8,12\}\), as described in \cref{section:benchmarks}. The table below reports one selected row per dimension; the costs are ranking scores within a fixed \(n\), not a cross-dimension comparison. For \(n \in \{8,12\}\), the claim is lexicographic minimality among the deterministic sample, not global optimality over all rows.
\end{enumerate}
The procedure produces the following generators:
\begin{center}
\small
\begin{tabular}{@{}c l@{}}
\toprule
\(n\) & generator row \(c\) \\
\midrule
4 & \((1, 1, 2, 3)\) \\
6 & \((1, -3, 1, 3, 2, 2)\) \\
8 & \((1, 8, 4, 5, 3, 12, 8, 7)\) \\
12 & \((12, 63, 65, 47, 72, 74, 20, 57, 59, 49, 64, 58)\) \\
\bottomrule
\end{tabular}
\end{center}

\paragraph{MDS Check.}
During the search, candidates are first tested over Mersenne31 using the incremental minor algorithm of Malakhov~\cite{Malakhov2110}. For a circulant matrix, this algorithm exploits cyclic symmetry to compute the \(k\)-minor layer from the \((k-1)\)-minor layer, avoiding the full \(\binom{n}{k}^2\) determinant enumeration. The first vanishing minor short-circuits the test. The selected rows are then re-verified over all target fields: Mersenne31, BabyBear, KoalaBear, and Goldilocks. This final pass catches rows that are MDS over one prime but not another. The implementation details and the FFT-block schedules used for the larger dimensions are discussed in \cref{section:benchmarks}; the algebra behind those schedules is given in \cref{appendix:mds-implementation}.

\section{The \design{} permutation and hash function}
\label{section:design}
We now formally specify the \design{} permutation that we define by utilising the functions described in \cref{subsection:non-linear-building-blocks}. 

%\subsection{Specification}
We define the \design{} permutation over \(\F_p^{2l}\), for some prime \(p\) and some positive integer \(l\). 
We provide a full specification for selected combinations within a set of supported primes \(p \in \left\{2^{31} - 1, 2^{31} - 2^{24} + 1, 2^{31} - 2^{27} + 1, 2^{64} - 2^{32} + 1\right\}\) and supported lengths \(l \in \left\{4, 6, 8, 12\right\}\).

\begin{remark}
A specification of \design{} over other parameter combinations is of course possible, but would need to be investigated for cryptographic security.
As this is not subject of this work, we leave it as an interesting research topic.
\end{remark}

\begin{definition}[\design{} Nonlinear Transformation]
\label{defn:design-non-linear-block-design}
Let \(p\) be a supported prime and \(\benes = \mathcal{B}_{f_0, f_1, f_2, f_3}\) the generic \beneswritten{} construction.
If \(\gcd(p^2-1, 5) = 1\), we fix \(f_0(x) = x^5\), \(f_1(x) = 5x^3 - 5x\), and \(f_3 = f_2 = f_0\). 
Otherwise, we fix \(f_0(x) = x^7 + x^3 + x\), \(f_1(x) = x^3 + x\), \(f_2(x) = x^7\), and \(f_3 = f_2\).
For a supported length \(\ell \), the nonlinear transformation in \design{} is $\nllayer: \F_p^{2l} \to \F_p^{2l}$ such that:
\[
    \left(x_0, \dots, x_{2l-1}\right) \mapsto \concat_{i=0}^{l-1} \benes\left(x_{2i}, x_{2i+1}\right).
\]
\end{definition}

Note that the two different choices of the \(f\) functions are such that \cref{prop:benes-invert}'s requirements are satisfied for all supported primes, so that \(\nllayer\) is indeed a permutation of \(\F_p^{2l}\).
We can easily distinguish the two settings by looking at \(d = \deg\left(f_0\right)\).

\begin{definition}[\design{} Linear Transformation]
    Let \(p\) be a supported prime and \(l\) a supported length, \(\theta \) some arbitrary imaginary unit for \(\mathbb{F}_{p^2}\), and \(M \in \F_{p^2}^{l \times l}\) the MDS matrix constructed in \cref{section:designrationale}.
    Let \(\Psi\left(x_0, x_1\right) = x_0 + \theta{}x_1\) be an embedding of \(\F_p^2\) into \(\F_{p^2}\), and \(\mathcal{E}\left(x_0, x_1, \dots, x_{2l - 2}, x_{2l - 1}\right) = \concat_{i=0}^{l-1} \Psi\left(x_{2i}, x_{2i+1}\right)\) its natural extension from \(\mathbb{F}_p^{2l}\) into \(\mathbb{F}_{p^2}^l\).
    The linear transformation in \design{} is $\mathcal{M}: \F_p^{2l} \to \F_p^{2l}$, such that:
    \[
   \left(x_0, \dots, x_{2l-1}\right) \mapsto \mathcal{E}^{-1}\left(M \cdot \mathcal{E}\left(x_0, \dots, x_{2l-1}\right)\right).
    \]
\end{definition}

Observe that the embedding is clearly a bijection, and similarly \(M\) is an invertible matrix as it is MDS, implying that \(\mathcal{M}\) is also a permutation of \(\F_p^{2l}\).

\begin{definition}[\design{} Round Function]
    \label{defn:round-function}
    Let \(p\) be a supported prime, \(l\) a supported length, and \(\mathbf{c} \in \mathbb{F}_p^{2l}\) some vector of constants. The round function in \design{} (with round constant \(\mathbf{c}\)) is $\roundfunction_{\mathbf{c}}: \F_p^{2l} \to \F_p^{2l}$, such that:
    \[
    \mathbf{x} \mapsto \llayer\left(\nllayer\left(\mathbf{x} + \mathbf{c}\right)\right).
    \]
\end{definition}
Clearly, \(\roundfunction_{\mathbf{c}}\) is a permutation of \(\F_p^{2l}\) for any choice of \(\mathbf{c}\).
An example of a 4-input round function is given in \cref{figure:round-func}.

\begin{figure}
    \centering
    \begin{tikzpicture}[
    scale=0.8,
    >=stealth,
    fbox/.style={draw, minimum width=0.5cm, minimum height=0.5cm, font=\scriptsize},
    mdsblock/.style={draw, rounded corners=2pt, minimum width=7.5cm, minimum height=0.6cm, font=\small},
    every node/.style={font=\small}
]

% === Column positions (two Benes blocks) ===
\def\La{0}
\def\Ra{2.5}
\def\Lb{5.5}
\def\Rb{8}

% === Brace amplitude in cm (5pt ~ 0.18cm) ===
\def\braceAmp{0.18}

% === Inputs ===
\node (x0) at (\La, 0) {$x_0$};
\node (x1) at (\Ra, 0) {$x_1$};
\node (x2) at (\Lb, 0) {$x_2$};
\node (x3) at (\Rb, 0) {$x_3$};

% === Round constant addition ===
\def\constY{-0.7}
\node (c0) [inner sep=0pt] at (\La, \constY) {\Large$\boxplus$};
\node (c1) [inner sep=0pt] at (\Ra, \constY) {\Large$\boxplus$};
\node (c2) [inner sep=0pt] at (\Lb, \constY) {\Large$\boxplus$};
\node (c3) [inner sep=0pt] at (\Rb, \constY) {\Large$\boxplus$};

\node (cc0) at (\La-1, \constY) {$c_{0}$};
\node (cc1) at (\Ra-1, \constY) {$c_{1}$};
\node (cc2) at (\Lb-1, \constY) {$c_{2}$};
\node (cc3) at (\Rb-1, \constY) {$c_{3}$};

\draw[->] (cc0) --(c0);
\draw[->] (cc1) --(c1);
\draw[->] (cc2) --(c2);
\draw[->] (cc3) --(c3);

\draw[->] (x0) -- (c0);
\draw[->] (x1) -- (c1);
\draw[->] (x2) -- (c2);
\draw[->] (x3) -- (c3);

% ================================================================
% === Benes block 1 (full internal structure) ===
% ================================================================
\def\bStartY{-1.4}
\def\fInnerY{-2.15}  % f_1, f_2 y-position
\def\fOuterY{-2.6}   % f_0, f_3 y-position

% Function boxes
\node[fbox] (F1a) at (\La, \fOuterY)       {$f_0$};
\node[fbox] (F2a) at (\La+0.8, \fInnerY)   {$f_1$};
\node[fbox] (F3a) at (\Ra-0.8, \fInnerY)   {$f_2$};
\node[fbox] (F4a) at (\Ra, \fOuterY)       {$f_3$};

% Addition nodes
\def\addYa{-3.4}
\node (adda0) [inner sep=0pt] at (\La, \addYa) {\Large$\boxplus$};
\node (adda1) [inner sep=0pt] at (\Ra, \addYa) {\Large$\boxplus$};

% Wiring: inputs to function boxes
\draw[->] (c0) -- (\La, \bStartY) -- (F1a);
\draw[->] (c0) -- (\La, \bStartY) -- (\La+0.8, \bStartY) -- (F2a);
\draw[->] (c1) -- (\Ra, \bStartY) -- (F4a);
\draw[->] (c1) -- (\Ra, \bStartY) -- (\Ra-0.8, \bStartY) -- (F3a);

% Wiring: function boxes to additions (crossing: f1->right add, f2->left add)
\draw[->] (F1a) -- (adda0);
\draw[->] (F3a.south) -- ++(0,-0.25) -- ($(adda0)+(0.4,0)$) -- (adda0);
\draw[->] (F4a) -- (adda1);
\draw[->] (F2a.south) -- ++(0,-0.25) -- ($(adda1)+(-0.4,0)$) -- (adda1);

% ================================================================
% === Benes block 2 (full internal structure) ===
% ================================================================

% Function boxes
\node[fbox] (F1b) at (\Lb, \fOuterY)       {$f_0$};
\node[fbox] (F2b) at (\Lb+0.8, \fInnerY)   {$f_1$};
\node[fbox] (F3b) at (\Rb-0.8, \fInnerY)   {$f_2$};
\node[fbox] (F4b) at (\Rb, \fOuterY)       {$f_3$};

% Addition nodes
\node (addb0) [inner sep=0pt] at (\Lb, \addYa) {\Large$\boxplus$};
\node (addb1) [inner sep=0pt] at (\Rb, \addYa) {\Large$\boxplus$};

% Wiring: inputs to function boxes
\draw[->] (c2) -- (\Lb, \bStartY) -- (F1b);
\draw[->] (c2) -- (\Lb, \bStartY) -- (\Lb+0.8, \bStartY) -- (F2b);
\draw[->] (c3) -- (\Rb, \bStartY) -- (F4b);
\draw[->] (c3) -- (\Rb, \bStartY) -- (\Rb-0.8, \bStartY) -- (F3b);

% Wiring: function boxes to additions (crossing: f1->right add, f2->left add)
\draw[->] (F1b) -- (addb0);
\draw[->] (F3b.south) -- ++(0,-0.25) -- ($(addb0)+(0.4,0)$) -- (addb0);
\draw[->] (F4b) -- (addb1);
\draw[->] (F2b.south) -- ++(0,-0.25) -- ($(addb1)+(-0.4,0)$) -- (addb1);

% ================================================================
% === Input curly braces (merge two wires into one arrow) ===
% ================================================================
\def\braceTopY{-3.9}

% Wires from additions down to brace endpoints
\draw (adda0.south) -- (\La, \braceTopY);
\draw (adda1.south) -- (\Ra, \braceTopY);
\draw (addb0.south) -- (\Lb, \braceTopY);
\draw (addb1.south) -- (\Rb, \braceTopY);

% Curly braces (mirror = opening downward)
\draw[decorate, decoration={brace, mirror, amplitude=\braceAmp cm}]
    (\La, \braceTopY) -- (\Ra, \braceTopY);
\draw[decorate, decoration={brace, mirror, amplitude=\braceAmp cm}]
    (\Lb, \braceTopY) -- (\Rb, \braceTopY);

% Single arrows from brace tips to MDS
\def\mdsY{-4.9}
\node[mdsblock] (M) at ({(\La+\Rb)/2}, \mdsY) {$M \in \F_{p^2}^{2 \times 2}$};

\pgfmathsetmacro{\braceTipAy}{\braceTopY-\braceAmp}
\draw[->] ({(\La+\Ra)/2}, \braceTipAy) -- ({(\La+\Ra)/2}, \braceTipAy |- M.north);
\draw[->] ({(\Lb+\Rb)/2}, \braceTipAy) -- ({(\Lb+\Rb)/2}, \braceTipAy |- M.north);

% ================================================================
% === Output curly braces (split one wire into two arrows) ===
% ================================================================
\def\braceBotYbase{-5.65}

% Curly braces (tip points up)
\draw[decorate, decoration={brace, amplitude=\braceAmp cm}]
    (\La, \braceBotYbase) -- (\Ra, \braceBotYbase);
\draw[decorate, decoration={brace, amplitude=\braceAmp cm}]
    (\Lb, \braceBotYbase) -- (\Rb, \braceBotYbase);

% Single wires from MDS down to brace tips
\pgfmathsetmacro{\braceBotTipY}{\braceBotYbase+\braceAmp}
\draw ({(\La+\Ra)/2}, \braceBotTipY |- M.south) -- ({(\La+\Ra)/2}, \braceBotTipY);
\draw ({(\Lb+\Rb)/2}, \braceBotTipY |- M.south) -- ({(\Lb+\Rb)/2}, \braceBotTipY);

% Two arrows from brace ends to outputs
\def\outY{-6.35}
\node (y0) at (\La, \outY) {$y_0$};
\node (y1) at (\Ra, \outY) {$y_1$};
\node (y2) at (\Lb, \outY) {$y_2$};
\node (y3) at (\Rb, \outY) {$y_3$};

\draw[->] (\La, \braceBotYbase) -- (y0);
\draw[->] (\Ra, \braceBotYbase) -- (y1);
\draw[->] (\Lb, \braceBotYbase) -- (y2);
\draw[->] (\Rb, \braceBotYbase) -- (y3);

\end{tikzpicture}
    \caption{Example of a round function (see \cref{defn:round-function}) with 4 inputs (\(l = 2\)). The nonlinear layer applies two \beneswritten{} blocks, whose outputs are paired into \(\F_{p^2}\) elements and mixed by an MDS matrix. The inputs, the round constants and the outputs are denoted as $x_i, c_i, y_i$ (for $i = 0, \dots, 3$), respectively.}
    \label{figure:round-func}
\end{figure}
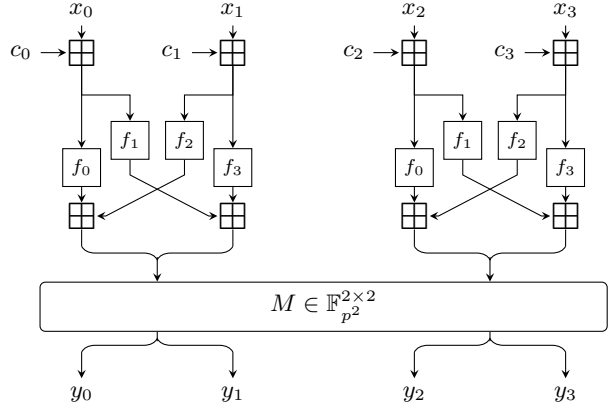

\begin{definition}[The \design{} permutation]\label{defn:design}
    Let \(p\) be a supported prime, \(l\) a supported length, \(R \in \mathbb{N}\) a number of rounds, and \(\mathbf{c}_1, \dots, \mathbf{c}_R \in \F_p^{2\ell}\) distinct round constants. 
    The \design{} permutation (with round constants \(\mathbf{c}_1, \dots, \mathbf{c}_R\)) is the map:
    \begin{equation*}
        \design_{\mathbf{c}_1, \dots, \mathbf{c}_R}^\pi = \roundfunction_{\mathbf{c}_R} \circ \dots \circ \roundfunction_{\mathbf{c}_1} \circ \mathcal{M} \circ \mathcal{L},
    \end{equation*}
    where
    \(\mathcal{L}(x_0, x_1, \dots, x_{2l-1}) = \concat_{i=0}^{l-1} \ell\left(x_{2i}, x_{2i+1}\right)\)  is the pairwise application of the Cooley--Tukey butterfly \(\ell(x_i, x_{i+1}) = (x_i + x_{i+1}, x_i - x_{i+1})\).
\end{definition}
It's immediate to see that \(\design{}^{\pi}\) is a permutation of \(\F_p^{2l}\) for any choice of round constants.

\paragraph{Choosing the Number of Rounds and the Constants.}
The recommended number of rounds for a given field and permutation size relies on the security analysis we carry out in \cref{section:securityanalysis}. We summarize our results for the supported parametrizations in \cref{table:roundnumbers}.
The $R \cdot t$ round constants are derived by rejection sampling from SHAKE256~\cite{shake:FIPS15} applied to the domain-separation ASCII-encoded string
\[
  \texttt{BENES-HASH/v1|p=}\langle p \rangle,
\]
where $\langle p \rangle$ is the field order encoded as a little-endian integer. Each constant is sampled by reading $\lceil \log_2 p \rceil / 8$ bytes from the SHAKE256 stream and accepting the resulting little-endian integer if it lies in $\{0, \ldots, p-1\}$, discarding candidates outside this range.

\begin{table}
  \centering
  \setlength{\tabcolsep}{10pt}
  \caption{Parameter summary for the \design{} permutation, targeting a security level \(\kappa = 128\), where \(2l\) is the state width, \(R\) is the number of rounds, and \(d\) is the degree of the round function. The last four columns specify the polynomials used in the \beneswritten{} construction.}\label{table:roundnumbers}
      \begin{tabular}{ccccccc}
    \toprule
    \(p\) & \(2l\) & \(R\) & \(d\) & $f_0$ & $f_1$ & $f_2 = f_3$\\
    \midrule
    \multirow{2}{*}{\(2^{31} - 1\)} & \(16\) & \(13\) & \multirow{4}{*}{\(5\)} & \multirow{4}{*}{$x^5$} & \multirow{4}{*}{$5x^3-5x$} & \multirow{4}{*}{$x^5$} \\
                                & \(24\) & \(13\) \\
                               \cmidrule{1-3}
    \multirow{2}{*}{\(2^{31} - 2^{24} + 1\)} & \(16\) & \(13\) \\
                                & \(24\) & \(13\) \\
                               \midrule
    \multirow{2}{*}{\(2^{31} - 2^{27} + 1\)} & \(16\) & \(11\) & \multirow{4}{*}{\(7\)} & \multirow{4}{*}{$x^7 + x^3 + x$} & \multirow{4}{*}{$x^3+x$} & \multirow{4}{*}{$x^7$} \\
                               & \(24\) & \(11\) \\
                              \cmidrule{1-3}
    \multirow{2}{*}{\(2^{64} - 2^{32} + 1\)} & \(8\) & \(22\) \\
                                 & \(12\) & \(22\) \\              
    \bottomrule
  \end{tabular}
\end{table}

In the remainder, we will drop the round constants from the subscript of \(\design^{\pi}\) when not relevant from the context, and we will assume the number of rounds follows our recommendations.

We propose two ways to instantiate the \design{} permutation into a compression function.

\begin{definition}[\design{} Sponge compression]\label{defn:design_sponge}
    Let \(p\) be a supported prime, \(l\) a supported length, and \(\design^{\pi}\) the corresponding \design{} permutation.
    Additionally, let \(c < 2l\) be an even positive integer, and \(h \le r\) be a positive integer, where \(r = 2l - c\).
    The \design{} Sponge compression function is the map:
    \begin{equation*}
        \design{}^S\colon \F_p^r \to \F_p^h; \qquad 
        \left(x_0, \dots, x_{r - 1}\right) \mapsto \floor{\design{}^{\pi}\left(x_0, \dots, x_{r - 1} \parallel \mathbf{0}^{c}\right)}_h
    \end{equation*}
    where \(\floor{\cdot}_n\) on a tuple denotes keeping its first \(n\) elements.
\end{definition}

\begin{definition}[\design{} feedforward compression]\label{defn:design_compression}
    Let \(p\) be a supported prime, \(l\) a supported length, and \(\design^{\pi}\) the corresponding \design{} permutation.
    Additionally, let \(h \le 2l\) be a positive integer.
    The \design{} feedforward compression function is the map:
    \begin{equation*}
        \design{}^C\colon \F_p^{2l} \to \F_p^h; \qquad 
        \mathbf{x} \mapsto \floor{\design{}^{\pi}\left(\mathbf{x}\right) + \mathbf{x}}_h
    \end{equation*}
    where \(\floor{\cdot}_n\) on a tuple denotes keeping its first \(n\) elements.
\end{definition}

\paragraph{Choosing the Sponge Parameters.}
The black-box security analysis~\cite{BertoniDPA2008,KhovratovichBM2023} of the Sponge construction requires choosing the capacity \(c\) and the digest size \(h\) such that \(p^h \gtrsim 2^{2\kappa}\) and \(p^c \gtrsim 2^{2\kappa}\) to achieve \(\kappa \) bits of random-oracle indifferentiability. 
Our security analysis supports a similar bound when considering the Constrained-Input-Constrained-Output (CICO) security model (see \cref{section:securityanalysis}). 
Concretely, for \(\kappa = 128\) bits, over the supported \(31\)-bit fields we recommend choosing the \design{} permutation corresponding to \(2l = 24\). Similarly, for the supported \(64\)-bit fields, we recommend choosing \(2l = 12\) and \(c = 4\) (see \cref{table:roundnumbers}). 
For the hash digest size, in all cases we recommend fixing \(h = r = 2l - c\), which matches the required bound.

\paragraph{Choosing the Feedforward Parameters.}
The black-box security analysis~\cite{AndreevaBRT2026} of the feedforward construction requires choosing the digest size \(h\) such that \(p^h \gtrsim 2^{2\kappa}\), to achieve \(\kappa \) bits of collision resistance.
In this setting, we always recommend choosing \(h = l\) for a target security level. In particular, for \(\kappa = 128\), we recommend choosing \(2l = 16\) over the supported \(31\)-bit fields and \(2l = 8\) over the supported \(64\)-bit fields.

\section{Security Analysis}
\label{section:securityanalysis}

\subsection{Security Goals and Adversary Model}
\label{subsec:security-goals}
We study an unkeyed permutation over \(\F_p^{2l}\) used as the permutation inside a Sponge-based hash function. The adversary has oracle access to the hash construction (or the permutation when relevant) and aims to find collisions, second preimages, or preimages, or to distinguish the construction from an idealized random oracle or permutation, within the parameter ranges of interest. We do not consider side-channel leakage or malicious implementations. Our analysis focuses on statistical and algebraic attacks, which are particularly relevant for circuit-friendly designs.

\subsection{Statistical Cryptanalysis}
\label{subsection:statisticalattacks}
We follow the adversary model in \cref{subsec:security-goals} and analyze the statistical properties of the permutation as an unkeyed hash function. Unlike a keyed primitive, we must ensure collision and second-preimage resistance, and also preimage resistance when the permutation is used inside a Sponge. For instance, for a statistically weak permutation in a Sponge, an attacker may find two inputs \(x_0 \neq x_1\) such that a (truncated) differential trail cancels or moves the difference in the non-squeezing part of the output.\footnote{Here we assume that each round function of our permutation is invertible.}

Moreover, \design{} works over \(\mathbb F_p\) for a prime \(p\). We consider truncated differentials only in terms of active nonlinear functions over the base field, and we assume the nonlinear building blocks \(f_0\) and \(f_1\) have no nontrivial linear structures.\footnote{The only subspaces of \(\mathbb F_p\) are \(\{0\}\) and itself.}

\subsubsection{Differential cryptanalysis.}
Differential cryptanalysis \cite{DBLP:conf/crypto/BihamS90} exploit non-random behavior by tracking how input differences propagate through a primitive. The attacker identifies high-probability differential trails, where each trail specifies a sequence of input and output differences across the rounds, and the overall probability is the product of the individual transition probabilities.

To give strict security arguments against this attack, we evaluate the minimum number of active nonlinear components across pairs of rounds using the wide trail strategy. In our design the atomic nonlinear component is the \beneswritten{} block \(\benes: \F_p^2 \mapsto \F_p^2\), and diffusion is provided by a linear layer using an MDS matrix.

For a state of \(2l\) base-field elements (i.e., \(l\) \beneswritten{} blocks per round), we evaluate two linear-layer approaches, namely a \(2l \times 2l\) MDS matrix over \(\F_p\) and an \(l \times l\) MDS matrix over the extension field \(\F_{p^2}\).

\begin{remark}
Note that, since the fields we are working with are closed, any MDS matrix over \(\F_p\) is also an MDS matrix over \(\F_{p^2}\). This makes it easier to choose efficient MDS matrices with small coefficients also in the \(\F_{p^2}\) case.
\end{remark}

\paragraph{Differential Probability of the \beneswritten{} Block.}
Before analyzing the diffusion layer, we establish the maximum differential probability of a single active \beneswritten{} block, assuming the nonlinear functions \(f_i\) have maximum degree \(d\).

\begin{lemma}
\label{lemma:benes_block_diff_prob}
Let \(\benes(x_0, x_1) = (f_0(x_0) + f_2(x_1), f_1(x_0) + f_3(x_1))\) be the \beneswritten{} block mapping over \(\F_p^2\). For any nonzero input difference \((\Delta x_0, \Delta x_1) \neq (0,0)\) and any nonzero output difference \((\Delta y_0, \Delta y_1) \neq (0,0)\), the differential probability is bounded from above by \(\frac{d-1}{p}\).\footnote{Since \beneswritten{} permutes, any nonzero input difference will lead to a nonzero output difference.}
\end{lemma}
We refer to \cref{appendix:proofs-statistical} for the proof.

\paragraph{Extension-Field MDS Matrix.}
Here we treat each pair of \(\F_p\) elements entering and exiting a \beneswritten{} block as a single element of the extension field \(\F_{p^2}\), and the linear layer is an \(l \times l\) MDS matrix over \(\F_{p^2}\).

\begin{proposition}
\label{prop:mds_ext_field}
A \(2\)-round construction with an \(l \times l\) MDS matrix over \(\F_{p^2}\) activates at least \(l+1\) \beneswritten{} blocks.
\end{proposition}
We refer to \cref{appendix:proofs-statistical} for the proof.

\subsubsection{Truncated Differentials.}
Truncated differential attacks \cite{DBLP:conf/fse/Knudsen94} follow trails of known and unknown differences without necessarily relying on the precise difference or probability. In our construction they are directly subsumed by the differential analysis above. Since the linear layer is an MDS matrix over \(\F_{p^2}\) with branch number \(l+1\), any truncated differential trail over two rounds must activate at least \(l+1\) \beneswritten{} blocks, whether the attacker tracks exact differences or only activity patterns. The MDS property in particular rules out the sparse activity patterns (e.g., a single active word propagating with probability \(1\) for several rounds) typically exploited in designs with weaker diffusion layers. Consequently, the bound in \cref{prop:mds_ext_field} already covers the worst case for truncated differentials.

\subsubsection{Rebound Attacks.}
Rebound attacks were introduced in 2009 and have since been widely used to analyze hash functions against various collision attacks, including combinations of statistical and algebraic approaches \cite{DBLP:conf/fse/MendelRST09,DBLP:conf/asiacrypt/LambergerMRRS09,DBLP:conf/fse/JeanNP12}. The attacker first defines a middle part of the permutation, called the inbound phase, often computed as a precomputation, in which a \enquote{suitable} internal state is sought. From this state the attacker uses differentials connecting it backwards to the start and forwards to the end of the permutation with non-negligible probability, which is the outbound phase. This may yield different input states producing a collision in the permutation output.

More formally, a rebound distinguisher on \design{} would be applied as follows.

\paragraph{Inbound Phase.} We specify a middle part of \(R_1\) rounds. The attacker exploits the low algebraic degree of the \beneswritten{} blocks to find a pair of internal states satisfying a prescribed differential transition
\begin{equation*}
(\alpha_0, \dots, \alpha_{2l-1}) \xrightarrow{R_1 \text{ rounds}} (\beta_0, \dots, \beta_{2l-1}).
\end{equation*}
Between consecutive nonlinear layers, the MDS matrix over \(\F_{p^2}\) ensures that any nonzero difference activates at least \(l+1\) \beneswritten{} blocks over every pair of rounds (see \cref{prop:mds_ext_field}). Therefore, even a single inbound round pair forces the attacker to satisfy at least \(l+1\) independent nonlinear constraints, each of degree \(d-1\). For \(R_1 = 2\) the resulting system of \(l+1\) degree-\((d-1)\) equations admits at most \((d-1)^{l+1}\) solutions, and extending the inbound phase further multiplies this cost.

\paragraph{Outbound Phase.} Starting from the inbound differentials, the attacker propagates backwards and forwards.
\begin{equation*}
    (\alpha_0, \dots, \alpha_{2l-1}) \xrightarrow{R_0 \text{ inverse rounds}} \Delta^\mathrm{(in)}, \qquad
    (\beta_0, \dots, \beta_{2l-1}) \xrightarrow{R_2 \text{ rounds}} \Delta^\mathrm{(out)}.
\end{equation*}
Each pair of outbound rounds activates at least \(l+1\) \beneswritten{} blocks via the MDS branch number, and each active block contributes a factor of at most \(\frac{d-1}{p}\) to the trail probability. After \(R_0\) backward rounds and \(R_2\) forward rounds, the outbound probability is therefore bounded by
\begin{equation*}
    \Pr[\text{outbound}] \leq \left(\frac{d-1}{p}\right)^{(l+1) \cdot \lfloor (R_0 + R_2)/2 \rfloor}.
\end{equation*}

\paragraph{Complexity.}
The total attack complexity is the algebraic cost of the inbound phase (solving a high-degree system over \(\F_p\)) combined with the statistical cost of the outbound phase. Unlike designs with weak diffusion layers (e.g., circular shifts), where probability-\(1\) truncated trails can extend over several rounds and enlarge the inbound phase for free, the MDS matrix in \design{} prevents this, since every round-pair transition implies a minimum number of active \beneswritten{} blocks. The inbound phase is therefore limited to very few rounds before the algebraic matching cost becomes prohibitive, while the outbound probability decreases exponentially in \((l+1)\) per two rounds. For the prime sizes and round counts targeted here, we do not expect rebound distinguishers to be a threat.

\subsubsection{Algebraic Degree Growth and Higher-Order Differentials.}
Higher-order differential and related attacks often exploit low algebraic degree together with rich families of affine subspaces, as is common over \(\mathbb F_{2^k}\). Over prime fields \(\mathbb F_p\) the only \(\mathbb F_p\)-linear subspaces are \(\{0\}\) and \(\mathbb F_p\) itself, so this setting provides no small structured subspaces to sum over, and exhausting one state word would require on the order of \(p\) inputs rather than the many combinatorial choices available over \(\mathbb F_{2^k}\). Combined with the repeated nonlinear permutation polynomials, we therefore expect the algebraic degree of the round function (as a multivariate polynomial map in the input words) to grow quickly with the rounds, and we do not expect higher-order differential distinguishers to be effective.

\subsubsection{Linear cryptanalysis.}
Linear cryptanalysis \cite{DBLP:conf/eurocrypt/Matsui93} exploits non-random correlations between linear (or affine) expressions in the input and output of a round function. Our permutation operates over a large prime field \(\mathbb F_p\) and uses low-degree permutation polynomials as its nonlinear layer, and we do not expect these components to exhibit exploitable linear structures or biases.

In a prime-field setting, linear cryptanalysis generalizes via additive-character correlations~\cite{DBLP:conf/sacrypt/BaigneresSV07,DBLP:conf/eurocrypt/Dobraunig0GK21}. For our nonlinear block \(\benes(x,y)=(f_0(x) + f_2(y), f_1(x) + f_3(y))\), any linear output mask expands into a sum of an \(x\)-only and a \(y\)-only polynomial, so the correlation factors into a product of two one-dimensional correlations. With our low-degree \(f_i\), these are expected to be very small unless the masked polynomial degenerates to an affine function, and the only easy cancellation eliminates the \(y\)-part while still leaving a nonlinear \(x\)-part. Over multiple rounds the MDS layer forces masks across many \(\benes\) applications, so any nontrivial trail accumulates many low-correlation nonlinear steps, making the overall bias negligible for the targeted prime sizes.

\subsubsection{Excluded Attacks.}
We propose \design{} as a hash function rather than a PRF, so we do not consider some specific attack vectors. In particular, we do not claim security against zero-sum partitions \cite{DBLP:conf/fse/BouraCC11} (which can use higher-order differentials \cite{DBLP:conf/fse/Knudsen94,DBLP:conf/crypto/BeyneCDELLNPSTW20} and/or integral/square attacks \cite{DBLP:conf/fse/DaemenKR97}), where the goal is to find disjoint sets of inputs and corresponding outputs that sum to zero. Further, we do not consider partial one-word collisions in the Sponge state. To the best of our knowledge, such a distinguisher cannot be turned into an attack on the hash or compression function, which follows the reasoning of similar work such as \cite{DBLP:conf/uss/0001KR0S21,DBLP:conf/crypto/GrassiHRSWW23}.

The same holds for other attacks such as impossible differential~\cite{DBLP:conf/eurocrypt/BihamBS99} and zero-linear correlation~\cite{DBLP:conf/fse/BogdanovW12}. The impossible differential attack exploits differentials holding with probability zero, and as for zero-sum partitions, we are not aware of any collision or (second) preimage attack based on impossible differential characteristics.

\subsubsection{Parameter Summary.}
For \(\kappa = 128\) bits of security and given the analysis above, two rounds of \design{} (see \cref{prop:mds_ext_field}) suffice to activate enough \beneswritten{} blocks in one direction. We add one more round for margin on full diffusion and two further rounds to counter attacks splitting the construction into parts and similar MitM approaches, leaving \(R_\mathrm{stat} = 5\) rounds against statistical attacks.

\subsection{Algebraic Cryptanalysis}
\label{subsection:algebraicattacks}
\subsubsection{Gr{\"o}bner Basis and Resultant-based Cryptanalysis.}

Gr{\"o}bner basis (GB) attacks, as well as Resultant attacks, are polynomial system solving techniques belonging to the elimination theory. Both techniques end by finding a univariate polynomial in some unknown input variable. The complexity of determining such polynomial depends on the quotient ring dimension \(D\) of the ideal defined by the system of equations describing the target cryptographic primitive. As this value acts as an invariant of such ideal, it can be used to estimate the complexity of the attack and, consequently, to determine the minimum number of rounds needed to reach a certain security level. 
Recent advances in cryptanalysis~\cite{DBLP:conf/crypto/BariantBLAOPR24,DBLP:journals/iacr/BakBBBOP25,DBLP:conf/asiacrypt/YangZYLT24,DBLP:conf/crypto/BariantBBHOR25,DBLP:conf/eurocrypt/CampaR25} of AO primitives have shown that determining the security of a design w.r.t.~the GB computation is not reliable as the given bound (e.g. Macaulay's bound) is not tight and the degree of regularity of the system is not easy to determine. 
Unlike GB computation complexity, procedures depending on the quotient ring dimension are well studied in the literature as well as their complexity, e.g. the complexity of computing the characteristic polynomial of a matrix of size \(D \times D\) is given as \(\bigO(D^\omega)\) where \(\omega\) is the algebra constant. Building upon all those results and considerations, from now on we assume the GB to be given, and we focus on the complexity of the univariate polynomial finding step or the corresponding root finding step.

The first step of the analysis is to consider the polynomial system describing the target cryptographic primitive within the context of the well-known CICO problem~\cite{report-stark-friendly}. Further details and the formal definition of the CICO problem can be found in~\Cref{appendix:ideal-degree-analysis}.

\paragraph{Univariate Polynomial Finding.}
As introduced, we assume computing the Gr{\"o}bner basis for \design{} is easy. The complexity of finding the univariate polynomial is based on the quotient ring dimension of the target ideal. In the case of zero-dimensional ideals, the degree of such a univariate polynomial is exactly the quotient ring dimension. Our approach to determine such a degree is described as: \textit{(i)} eliminate useless variables from the given polynomial system: we eliminate $n-k$ variables per round, thus obtaining $k$ equations per round, for a total of $Rk$ equations in $Rk$ variables; \textit{(ii)} show that the B{\'e}zout bound is tight, even with high probability, by showing that the final polynomial system has no solutions at infinity.

\begin{proposition}[\(D\) for the \(\cico\)-($n$-$k$,$k$) Problem]
\label{prop:quot-ring-dim-CICO}   
Let consider the \(\cico\)-($n$-$k$,$k$) problem applied to the \design{} permutation. Let \(n = 2l\) be the number of branches and let \(k\) denote the number of constrained outputs ($n-k$ inputs are constrained to obtain a determined polynomial system). Each round generates \(k\) polynomials of degree \(d \in \{5, 7\}\) (depending on the prime). With high probability, the quotient ring dimension is given as \(D = d^{kR}\), where \(R\) denotes the number of rounds. 
\end{proposition}
The proof is given in \Cref{appendix:ideal-degree-analysis}.
This result was further supported by experimental evidence obtained by computing the GB of the polynomial system modelling the CICO-($n$-$k$,$k$) problem for toy versions of \design{}.

\paragraph{Factorization or Root Finding.}

Once the univariate polynomials for the unknown inputs are found, we compute their roots to obtain preimages. Several root-finding algorithms over finite fields exist, e.g., Cantor--Zassenhaus~\cite{cantor1981new} and Kaltofen--Shoup~\cite{Shoup1993FactoringPO}, with complexity depending on the polynomial degree and the field size. For instance, Cantor--Zassenhaus has expected time \(\mathcal{O}(D \log^2 D \log \log D)\) for a degree-\(D\) polynomial over a finite field. Although full factorization yields all preimages, finding a single root often suffices to compromise preimage resistance. Single-root finding is generally cheaper than full factorization, often \(\mathcal{O}(D \log D)\) using probabilistic methods or by computing the GCD of the polynomial with \(x^{p} - x\) (the field equation), where \(p\) is the field size. We expect this GCD to have much smaller degree than the original polynomial, so factoring the final GCD is far more efficient. For a degree-\(D\) polynomial \(h(x)\) over \(\F_p\), computing \(\gcd(h(x), x^{p} - x)\) costs \(\mathcal{O}(D \log^2(D))\) with the Half-GCD algorithm.

For 64-bit prime fields the field equation has high degree. To avoid this, we use the observation of~\cite{von2003modern,DBLP:journals/tosc/GrassiKRS22} that \(\gcd(h(x), x^{p} - x) = \gcd((x^p \pmod{h(x)}), h(x))\) when \(\deg(h(x)) < p\). In particular, \((x^p \pmod{h(x)}) -x\) has degree much lower than \(p\) when \(\deg(h(x)) \ll p\).
The resulting complexity estimate, accounting for the polynomial multiplication to compute \(x^p\), the GCD computation, and the final factorization, is given as \(D\log^2(D) \cdot (\log(D) + \log(p))\cdot (1 + 63.43\cdot \log(\log(D))).\)

\paragraph{The Minimum Number of Rounds.}
The complexity of the univariate polynomial finding and the root-finding steps strongly depends on the quotient ring dimension. Due to recent advances in algebraic cryptanalysis and, in particular, on the usage of elimination-based methods (e.g. resultant attacks), we decided to take a conservative approach toward the estimation of the minimum number of rounds needed to reach a given security level. In particular, we decided to base our estimation on the probabilistic root-finding complexity, given as \(D\log(D)\) where \(D\) is the quotient ring dimension. Hence, the number of rounds \(R\) needed to reach a security level of \(\kappa\) bits is given as \[R_\mathrm{alg} = \min\{ R \mid d^{Rk}\log\left(d^{Rk}\right) \ge 2^\kappa \},\] where \(d\) is the degree of the nonlinear functions and \(k\) is the number of constraint outputs. Considering the parameters given in \cref{table:roundnumbers}, the formula would give $7$ rounds for Mersenne31/KoalaBear, $6$ for BabyBear and $11$ Goldilocks.

\subsubsection{Interpolation Attacks.}
A highly effective technique is the interpolation attack, introduced by Jakobsen and Knudsen in 1997~\cite{DBLP:conf/fse/JakobsenK97}. Given a keyed function $E_k(\vars{x}) : \F_p \mapsto \F_p$, the attack reconstructs a polynomial representation of $E_k$ without knowing the secret key. With such a polynomial, the adversary can evaluate the function on arbitrary inputs, enabling forgeries and related attacks. In the best case, the interpolation polynomial can be recovered with Fast Lagrange interpolation in $\mathcal{O}(D\log(D))$, where \(D\) is the number of required input-output pairs.
The attack extends naturally to unkeyed permutations $E(\cdot)$ when the interpolation polynomial can be derived without full codebook access, enabling forgery attacks infeasible for (pseudo-)random permutations. When the permutation is used in a Sponge construction, only partial outputs are observable, but polynomial reconstruction may still let the adversary compute preimages, breaking preimage resistance.

Each permutation output can be written as a multivariate polynomial in the input variables. In the worst case, when all but one input is fixed or known, each output reduces to a univariate polynomial in the remaining unknown. If its degree is low enough, it can be recovered from few input-output pairs via interpolation, after which it allows efficient output prediction or preimage computation using root-finding.
Hence, let's consider the degree of each output polynomial in the unknown input.

\begin{lemma}[Degree of the Output after \(R\) Rounds]
\label{lem:degree-after-r-rounds}
Let \(n\) be the number of branches (inputs) of the permutation and let \(d = \deg(f_0) = \deg(f_2) = \deg(f_3)\) be the maximum degree of the nonlinear functions. The degree of each output after \(R \ge 1\) rounds of the \design{} round function is given as \(d^{R}\).
\end{lemma}
The proof is given in \Cref{appendix:proofs-algebraic-degree}.

\begin{proposition}
    \label{prop:interpolation-attack-complexity}
    Let \(n\) be the number of branches (inputs) of the \design{} round function and let \(d = \deg(f_0) = \deg(f_2) = \deg(f_3)\) be the maximum degree of the nonlinear functions. Then, if all but one input is fixed, after \(R\) rounds each output can be represented as a univariate polynomial of degree \(D = d^{R}\) in terms of the unknown input word. Hence, the interpolation attack complexity is given as \(\mathcal{O}(D\log(D))\).
\end{proposition}

\subsubsection{Polynomial degree and Security of Hash Function}
\label{rem:interpolation-degree}
In a Sponge mode, a univariate output polynomials can be obtained by fixing all but one variable of the input to the \design{} permutation. A common approach in constructing AO permutation over $\F_p$, is to ensure that the degree of the univariate output polynomial is $\approx p-2$. When $p$ is large and the degree $d \lll p$ one can construct the output polynomials with $\mathcal{O}(d)$ queries. If $p$ is small, this can be feasibly done with $\mathcal{O}(p)$ queries, even if it is ensured by design that the
degree of the univariate polynomials has maximum possible degree. The construction of such univariate polynomial is not a weakness unless the corresponding
query complexity is less than $p-1$ in general. In a small field the situation is more subtle because both $p-1$ or $d$ $(< p-1)$ queries are practically feasible. Nevertheless, one can still \textit{choose} to keep the general criteria (even if it gives practical complexity).     

When the permutation is defined over $\F_p$ for small $p$, we must consider the security implication of the practical construction of such univariate polynomials. The primary concern is to ensure the security of a hash (or compression) function. Suppose we have $t$ such univariate polynomials, then finding the common root(s) of them implies finding preimage (or collision). When $t$ random polynomials are given, the probability that they will share one or more root(s) decides the security bound. Note that such common root(s) may reside in the closure of $\F_p$ in general. However, for a cryptographic security, we consider the probability that the common root(s) are in $\F_p$. Formally, for collision security we consider the event $N = \left| \{a \in \F_p: f_i = (a) = 0, 1 \leq i \leq t\}\right| \geq 2$. We show that using existing results \cite{Jain2024commonroot} on the distribution of common roots of random polynomials, it is possible to estimate the probability (bound) of $\text{Pr}[N \geq 2]$. 

\begin{theorem}[\cite{Jain2024commonroot}]
  Given $f_1, \ldots, f_t \xleftarrow{\$} V \subset \F_p[x_1, \ldots, x_n]$ sampled independently, where $f_i: \F_p^n \mapsto \F_p$, $N = \left| \{a \in \F_p: f_i(a) = 0, 1 \leq i \leq t\}\right|$ follows Binomial distribution $\text{Bin}(p^n, 1/p^t)$. 
\end{theorem}
Note that~\cite{Jain2024commonroot} proves this for an arbitrary finite field $\F_p$. We need a special case of this theorem, namely when $n = 1$ and $V = \F_p[x]_{\leq p-1}$. Then, from the above theorem, it follows that $N \sim \text{Bin}(p, p^{-t})$. Hence, we get 
\[
\prob[N \geq 2] = \frac{1}{2}p^{2-2t}(1 + \mathcal{O}(p^{-1}))
\]
for $t \geq 2$. For arbitrary $d < p-1$, the bound on $\prob[N \geq 2]$ is slightly different. For a uniformly random polynomial $f$ we have
$\prob[f(a) = f(b) = 0] = p^{-2}$ for any distinct $a, b \in \F_p$. Since $t$ polynomials are independent, $\prob[f_i(a) = f_i(b) = 0, \forall i] = p^{-2t}$, Applying union bound over $\binom{p}{2}$ pairs we get
\[
\prob[N \geq 2] \leq \binom{p}{2} \cdot p^{-2t}
\]
A special case occurs when $d=1$, since a non-zero polynomial with $d=1$  can not have more than one root. So, $N \geq 2$, if every polynomial $f_i$ is a zero polynomial. Hence, $\prob[N \geq 2] \leq p^{-2t}$.

Assuming that the output polynomials for \design{} behave as in the above discussion, the probability of finding collision is extremely low. Our experimental results with the \design{} show (for $31$-bit prime and reduced rounds -- 3 to 6) that we never find two common zeros corresponding to the output polynomials. This indicates that the actual probability is indeed extremely small for the choice of field. 
Recent cryptanalysis~\cite{sbox-skipping} of \poseidon{} also shows a behaviour consistent with the above description.

\subsection{Choosing the Number of Rounds}
\label{subsection:cryptanalysissummary}
Let \(R_\mathrm{stat}\) and \(R_\mathrm{alg}\) be the number of rounds required by our statistical and algebraic analysis respectively; the number of rounds would be given as \(R_\mathrm{min} = \ceil{\max(R_\mathrm{stat}, R_\mathrm{alg})}\). Considering the discussion in \Cref{rem:interpolation-degree}, in particular the fact that $R_\mathrm{min}  < \floor{\log_d p}$, we define 
\[
R = \floor{\log_d p} = \max\left\{ R \mid d^R \le p \right\},
\]
letting the degree \(d^R\) of \cref{lem:degree-after-r-rounds} grow until it saturates the field. 
This gives \(R = 13, 11, 22\) for KoalaBear/Mersenne31, BabyBear, and Goldilocks respectively. 
\cref{table:roundnumbers} summarizes the parameters of \design{} w.r.t. the chosen prime.

\section{Practical Evaluation}
\label{section:benchmarks}
\subsection{Hardware Performance}
To quantify where our hash function stands against other circuit-friendly hash functions, we implement it in four instantiations on an FPGA platform and compare two of them with their \poseidon{} and \poseidontwo{} counterpart. We implement \design{} for the Goldilocks, Mersenne31, BabyBear, and KoalaBear fields, and both \poseidon{} and \poseidontwo{} for the Goldilocks and BabyBear fields using the parameters from Plonky3\footnote{\url{https://github.com/Plonky3/Plonky3}}, a framework for generating Polynomial Interactive Oracle Proofs (PIOPs). We target the AMD Virtex UltraScale+ HBM VU47P FPGA available on AWS and use Vivado 2025.1 for synthesis and implementation.

\subsubsection{Performance Comparison.} The resource utilization and timing results obtained in Vivado's out-of-context mode are presented in \cref{tab:hash_fpga}. Resource utilization is reported as the number of Look-Up Tables (LUTs), the FPGA's configurable combinational logic units; Flip-Flops (FFs), its register bits; Digital Signal Processors (DSPs), hard-wired multipliers from which we build the field multipliers; and Block RAMs (BRAMs), larger but slower on-chip memories.

Within the \design{} family, the Mersenne31 instance achieves both the highest throughput and the lowest resource utilization. This is expected: its modular reduction is the most efficient of the four fields, so fewer pipeline stages are needed to reach a high clock frequency and fewer resources are used overall. The Goldilocks instance has the highest latency, but its throughput is relatively close to that of Mersenne31, again thanks to an efficient modular reduction. Its resource utilization, however, is significantly higher than that of the other instances, because the field requires 64-bit rather than 31-bit arithmetic, which substantially increases the DSP count. The KoalaBear instance reaches a higher clock frequency than the BabyBear instance, on par with Mersenne31, at the cost of more clock cycles and more resources. The BabyBear instance runs at a lower frequency because of its more complex modular reduction. As was done for KoalaBear, the frequency of the BabyBear instance could be increased by adding pipeline stages, at the cost of a higher LUT/FF utilization.

We first compare \design{} with \poseidon{}. For Goldilocks, \design{} uses considerably more resources, namely 69.3\% more LUTs, 64.0\% more FFs, and 104.7\% more DSPs, but it needs about 19.7\% fewer clock cycles. In terms of performance, \design{} achieves a 3.1\% higher steady-state throughput and, with a latency of 3.427 $\mu$s versus 4.401 $\mu$s, a 22.1\% lower latency than \poseidon{}.

For BabyBear, the two implementations have effectively the same throughput: 398.6 MHz for \design{} and 399.7 MHz for \poseidon{}. Here, \design{} uses 32.9\% fewer LUTs and 26.6\% fewer FFs, at the cost of 24.8\% more DSPs. It also reduces the cycle count from 1,241 to 729 and the latency from 3.105 $\mu$s to 1.829 $\mu$s, i.e., 41.3\% fewer cycles and a 41.1\% lower latency, which corresponds to a 1.7x improvement in latency.

We now compare \design{} with \poseidontwo{}. For Goldilocks, \design{} again has a higher resource utilization, requiring approximately 113.3\% more LUTs, 92.5\% more FFs, and 104.7\% more DSPs. In terms of throughput, \design{} is slightly faster: 478.5 MHz versus 455.8 MHz, i.e., about 5.0\% higher. In terms of latency, \design{} requires 1,640 cycles (3.427 $\mu$s), compared with 2,132 cycles (4.678 $\mu$s) for \poseidontwo{}, i.e., a 26.7\% lower latency, which corresponds to a 1.36x improvement in latency.

For BabyBear, the resource utilization gap between \design{} and \poseidontwo{} is much smaller: \poseidontwo{} uses 10.6\% fewer LUTs, 12.9\% fewer FFs, and 3.3\% fewer DSPs, and it reaches a 4.2\% higher steady-state throughput.
In return for this modest increase in resource usage and slight reduction in throughput, \design{} reduces the cycle count from 1,636 to 729 and the latency from 3.939 $\mu$s to 1.829 $\mu$s. This corresponds to 55.4\% fewer cycles and a 53.6\% lower latency, i.e., a 2.15x improvement in latency over \poseidontwo{}.

\begin{table}[t]
  \centering
  \caption{Resource utilization and performance comparison of \design{}, and \poseidontwo{} hash functions.}
  \label{tab:hash_fpga}
  
  \begin{threeparttable}
  \renewcommand{\TPTminimum}{\textwidth}% notes width: the tabular is scaled by \resizebox, which threeparttable cannot see
  \setlength{\tabcolsep}{6pt}
  \renewcommand{\arraystretch}{1.2}

  \resizebox{\textwidth}{!}{%
  \begin{tabular}{l|l|l|l|l|l}
  \Xhline{1.1pt}
  \makecell{Hash \\ Function} & Field & \makecell[l]{Resource Utilization\\(LUTs/FFs/DSPs/BRAMs)$^{\ddagger}$} & \makecell{Freq$^{\dagger}$ \\ {[MHz]}} & \makecell{Clock \\ Cycles} & \makecell{Time\\ {[$\mu$s]}} \\
  \Xhline{1.1pt}
  \multirow{4}{*}{\makecell{\design{}}} & Goldilocks & 1,104,100/1,414,041/8,448/0 & 478.5 & 1,640 & 3.427 \\
  \cline{2-6}  
   & BabyBear & 446,731/631,601/2,112/0 & 398.6 & 729 & 1.829 \\
  \cline{2-6}
   & KoalaBear & 548,419/719,796/1,872/0 & 501.0 & 1,162 & 2.319 \\
  \cline{2-6}
   & Mersenne31 & 266,039/418,369/1,872/0 & 506.1 & 634 & 1.253 \\
  \hline\hline
  \multirow{2}{*}{\poseidon{}} & Goldilocks & 652,054/862,090/4,128/0 & 464.0 & 2,042 & 4.401 \\
  \cline{2-6}
   & BabyBear & 665,775/860,382/1,692/0 & 399.7 & 1,241 & 3.105 \\
  \hline\hline
  \multirow{2}{*}{\poseidontwo{}} & Goldilocks & 517,661/734,573/4,128/0 & 455.8 & 2,132 & 4.678 \\
  \cline{2-6}
   & BabyBear & 399,479/550,316/2,043/0 & 415.3 & 1,636 & 3.939 \\
  \hline
  \Xhline{1.1pt}
  \end{tabular}%
  }
  \begin{tablenotes}[flushleft]
    \footnotesize
    \item[$\dagger$] The frequency also represents the throughput in MHashes/s.
    \item[$\ddagger$] Look-Up Tables (LUTs) are the FPGA's configurable combinational compute units, Flip-Flops (FFs) provide the memory registers, Digital Signal Processors (DSPs) are high-performance 27-bit by 18-bit signed multipliers from which we build 32-bit multipliers using Karatsuba, and 64-bit multipliers using the schoolbook approach on the 32-bit Karatsuba multipliers, and Block RAMs (BRAMs) are larger but slower on-chip memories.
  \end{tablenotes}
  \end{threeparttable}
  \end{table}

\subsubsection{Resource Sharing with the NTT.}
The previous discussion concerned standalone, fully unrolled hash cores. We now evaluate dual-mode designs in which an NTT and a two-round hash tile share a single physical pool of field multipliers. For this evaluation, we use one fully pipelined $2^{12}$-point multipath delay-commutator (MDC) NTT per field, with parallelism 16 for BabyBear and 8 for Goldilocks, which matches the number of field elements in the hash state. This NTT size is common in FHE and ZK hardware, and such an NTT can serve as a building block for larger ones. For each of the two fields, we implement three dual-mode variants, each pairing the NTT with a two-round tile of \design{}, \poseidon{}, or \poseidontwo{}; a mode selected at runtime switches between the NTT and the hash datapath. Since a tile implements only two rounds, it must be invoked several times to compute a full hash. For simplicity, we restrict the analysis to the full and partial rounds of \poseidon{} and \poseidontwo{}, and to the main round of \design{}, i.e., its round function \(\roundfunction\) (\cref{defn:round-function}).

Let $M_{\mathrm{dual}}$ denote the number of field multipliers in a dual-mode design and $M_{\mathrm{NTT}}$ the number of field multipliers required by the standalone NTT. We define the \emph{multiplier overprovisioning} factor as $F_{\mathrm{mul}}=M_{\mathrm{dual}}/M_{\mathrm{NTT}}$, i.e., the factor by which the multiplier pool grows to support the hash, and the corresponding NTT utilization as $U_{\mathrm{NTT}}=1/F_{\mathrm{mul}}$, i.e., the fraction of the pool that the NTT uses. Adding hash support increases the pool from $M_{\mathrm{NTT}}=96$ to $M_{\mathrm{dual}}=128$ field multipliers for BabyBear and from $M_{\mathrm{NTT}}=48$ to $M_{\mathrm{dual}}=64$ for Goldilocks. The DSP count grows by the same factor of one third (each 31-bit multiplier uses three DSPs and each 64-bit multiplier twelve): from 288 to 384 for BabyBear and from 576 to 768 for Goldilocks. In both fields, therefore, $F_{\mathrm{mul}}=4/3$ and $U_{\mathrm{NTT}}=75\%$. Compared with the sum of the resources of the standalone NTT and of the standalone hash tile, the shared designs save 42.9\% of the DSPs, 21.7--25.9\% of the LUTs, and 25.8--30.3\% of the FFs.

The size of the shared multiplier pool alone does not indicate how fully, or how efficiently, the pool is used in hash mode, because not all multipliers are necessarily computing at all times. We therefore define the steady-state multiplier saturation as $S=M_{\mathrm{active}}/M_{\mathrm{dual}}$, where $M_{\mathrm{active}}$ is the number of multipliers that are actively computing in steady state and $M_{\mathrm{dual}}$ is, as above, the total number of multipliers in the pool. The \design{} tile applies its nonlinear map to every lane (i.e., every element) of the state in both unrolled rounds and therefore uses the entire pool, giving $M_{\mathrm{active}}=M_{\mathrm{dual}}$ and \(S=100\%\).

For \poseidon{} and \poseidontwo{}, the saturation depends on whether the tile executes two full rounds or two partial rounds. Each \(x^7\) S-box costs four field multiplications, and a full round applies the S-box to every lane. Over two full rounds, the tile therefore keeps \(2\cdot16\cdot4=128\) multipliers active for BabyBear (two rounds, 16 lanes, four multiplications per S-box) and \(2\cdot8\cdot4=64\) for Goldilocks (two rounds, 8 lanes), i.e., the entire pool in both cases, giving \(S_{\mathrm{full}}=100\%\).

A partial round instead applies the S-box to the first lane only. Over two partial rounds, the S-boxes therefore occupy \(2\cdot1\cdot4=8\) multipliers of the pool, which yields \(S_{\mathrm{partial}}=8/128=6.25\%\) for BabyBear and \(S_{\mathrm{partial}}=8/64=12.5\%\) for Goldilocks when only the S-box layer is accounted for. For \poseidon{}, whose linear layer is multiplier-free, these are the final values. The Goldilocks \poseidontwo{} tile has the same partial-round saturation of 12.5\%, because its internal matrix is multiplier-free as well. The BabyBear \poseidontwo{} tile, in contrast, additionally requires nine generic field multiplications per partial round for its internal matrix. Over two partial rounds, its active multiplier count is therefore \(M_{\mathrm{active}}=2\cdot(4+9)=26\) (four S-box multiplications plus nine internal-matrix multiplications per round), which yields \(S_{\mathrm{partial}}=26/128\approx20.3\%\).
\subsection{Software Performance}
Besides our hardware evaluation, we benchmark \design{} in a plain Rust implementation for 64-bit platforms and compare it with the relevant literature. Many implementations of arithmetization-friendly hash functions exist, which gives implementers flexibility but makes a fair comparison difficult, since the supported prime fields, the implemented hash functions, and the level of dedicated optimization vary greatly. We therefore implement \design{} in different software frameworks for a nuanced view of its performance.

First, we implement \design{} in Plonky3, instantiated with KoalaBear, BabyBear, Mersenne31, and Goldilocks. We compare it with implementations of \rescueprime{}, \rpo{}, \poseidontwo{}, \poseidon{}, and \monolith{} for all primes where implementations exist. For \poseidon{}, the developers of Plonky3 use the round numbers suggested in the \poseidontwo{} paper~\cite{DBLP:conf/africacrypt/GrassiKS23}. For MDS matrices over small primes, they employ low-norm circulant matrices which can be evaluated with Karatsuba-style convolution in the full rounds. In the partial rounds, they use the sparse matrix decomposition from the original \poseidon{} paper~\cite[Appendix~B]{DBLP:conf/uss/0001KR0S21}, which reduces the cost of evaluating the linear layer to $O(t)$.

The hash functions available in Plonky3 span inverse power maps, low-degree nonlinear functions, and split-and-lookup layers. \cref{tab:comparison_software_plonky3} lists the average performance of a single permutation call for state sizes $16$ and $24$ over the 31-bit primes and $8$ and $12$ over the 64-bit prime.
As expected, \design{} clearly outperforms the inverse power map designs \rescueprime{} and \rpo{} for all supported state sizes and primes. At the other end of the spectrum, split-and-lookup designs such as \monolith{} excel at plain performance, yet \design{} is faster for the larger Mersenne31 state: while the evaluation time of the nonlinear layers of \monolith{} scales linearly with the state size, its MDS matrices are significantly more expensive at larger state sizes. Results for \monolith{} over 31-bit primes are limited to Mersenne31, since the decomposition step in its Bars layer depends on the prime and KoalaBear and BabyBear were unsupported at the time of writing. Among the power map designs, \design{} outperforms \poseidon{} in all implemented instances but is faster than \poseidontwo{} only for the 16-element BabyBear state.
Overall, \design{} performs solidly against state-of-the-art constructions.

An additional comparison over Goldilocks in the \enquote{ZK-friendly Hash Zoo} is provided in \cref{appendix:performance}. The library supports a broader selection of hash functions, though most small-prime instances are only available over the Goldilocks field.
\begin{table}[t]
  \centering
  \caption{Software performance of a single permutation call in the Plonky3 framework, for state sizes of approximately 512 and 768 bits.
  The reported results are given in nanoseconds, using 100 samples per benchmark.}\label{tab:comparison_software_plonky3}
  \resizebox{\linewidth}{!}{
  \begin{threeparttable}
\begin{tabular}{l|cc|cc|cc|cc}
\toprule
\multicolumn{1}{c}{} & \multicolumn{2}{c}{KoalaBear} & \multicolumn{2}{c}{BabyBear} & \multicolumn{2}{c}{Mersenne31} & \multicolumn{2}{c}{Goldilocks} \\
\cmidrule(lr){2-3}\cmidrule(lr){4-5}\cmidrule(lr){6-7}\cmidrule(lr){8-9}
\multicolumn{1}{c}{Hash} & ($t = 16$) & ($t = 24$) & ($t = 16$) & ($t = 24$) & ($t = 16$) & ($t = 24$) & ($t = 8$) & ($t = 12$) \\
\midrule
\rescueprime{} & 9\,226 & 14\,945 & 9\,535 & 15\,278 & 4\,571 & -- & 7\,433 & 11\,155 \\
\poseidon{}\tnote{a} & 1\,006 & 2\,550 & 1\,088 & 2\,786 & 827 & -- & 737 & 1\,111 \\
\poseidontwo{} & \textbf{602} & \textbf{880} & 759 & \textbf{1\,228} & 501 & \textbf{859} & 480 & 669 \\
\monolith{} & -- & -- & -- & -- & \textbf{349} & 1\,867 & \textbf{257} & \textbf{388} \\
\rpo{} & -- & 4\,927 & -- & 5\,147 & -- & 6\,211 & -- & 3\,836 \\
\midrule
\design{} & 921 & 1\,461 & \textbf{730} & 1\,376 & 671 & 971 & 580 & 1\,063 \\
\bottomrule
\end{tabular}
\begin{tablenotes}[flushleft]
  \footnotesize
  \item[a] For \poseidon{}, we use the round numbers suggested in the \poseidontwo{} paper~\cite{DBLP:conf/africacrypt/GrassiKS23}, since the original \poseidon{} instances were not designed for these small prime fields.
\end{tablenotes}
\end{threeparttable}}
\end{table}
\subsection{Proof System Performance}
As an additional benchmark, we implement and evaluate the arithmetization of \design{} in the Plonky3 framework, using its univariate STARK with a FRI-based 2-adic polynomial commitment scheme over the BabyBear, KoalaBear, and Goldilocks prime fields. Plonky3 provides an Algebraic Intermediate Representation (AIR) language to encode a computation trace into constraints. We pack the entire trace of 8 hash calls into a single row of the trace matrix and create $2^{12}$ such rows for a compact proof of $2^{15}$ hash evaluations. This structure batches many hash evaluations into a single proof to benchmark amortized performance, where proof-system performance is dictated primarily by the trace matrix size and the degree of intermediate algebraic relations.

We use two arithmetization variants of \design{}. The first is a direct implementation as an arithmetic circuit. The second introduces intermediate witnesses in each branch of the \beneswritten{} layer to reduce the maximum degree, which for both the function choice in our invertible \beneswritten{} networks brings the maximum degree down to three. The first arithmetization uses fewer trace cells and has reduced verification time, while the second one reduces prover time. We denote the small and fast variants by \design$_s$ and \design$_f$, respectively.
\cref{tab:proof_system_performance} gives the number of trace cells per permutation, amortized prover time, and verification latency for proofs of $2^{15}$ permutation calls of \design{} with its two variants, \poseidon{}, \poseidontwo{}, and \monolith{} for the selected primes and both state sizes. The intermediate witnesses moderately reduce prover time, at the cost of more trace cells and higher verification latency.
\poseidontwo{} outperforms both variants of \design{} in prover and verification time for every tested field and state size. \design$_s$ uses fewer trace cells than \poseidon{} and \poseidontwo{} in BabyBear, but more in KoalaBear and Goldilocks. In Goldilocks and BabyBear, \poseidon{} and \poseidontwo{} keep the maximum degree at 3 with one intermediate register per degree-7 power map; \design$_f$ uses two registers per S-box on average, matching the intermediate state count of a \poseidon{} full round while representing two bivariate polynomials in our \beneswritten{} network.

\begin{table}[h]
  \centering
  \caption{Proof system performance for proofs of $2^{15}$ permutation calls. Cells denotes trace cells per permutation; prover time is amortized per permutation, and verifier time is per proof. The upper and lower panels use $t=16$ and $t=24$ for KoalaBear and BabyBear, and $t=8$ and $t=12$ for Goldilocks. The first row of \design{} uses no intermediate witnesses; the second uses the fast arithmetization.}\label{tab:proof_system_performance}
\begin{minipage}{\linewidth}
  \bigskip
\textit{KoalaBear/BabyBear: $t=16$; Goldilocks: $t=8$}\par\smallskip
\resizebox{\linewidth}{!}{%
\begin{tabular}{lrrr rrr rrr}
\toprule
\multirow{3}{*}{Construction} & \multicolumn{3}{c}{KoalaBear} & \multicolumn{3}{c}{BabyBear} & \multicolumn{3}{c}{Goldilocks} \\
\cmidrule(lr){2-4}\cmidrule(lr){5-7}\cmidrule(lr){8-10}
 & \multirow{2}{*}{Cells} & Prover & Verifier & \multirow{2}{*}{Cells} & Prover & Verifier & \multirow{2}{*}{Cells} & Prover & Verifier \\
 &  & ($\mu$s/perm.) & (ms/proof) &  & ($\mu$s/perm.) & (ms/proof) &  & ($\mu$s/perm.) & (ms/proof) \\
\midrule
\multirow{2}{*}{\design{}} & 224 & 136 & 15.8 & \textbf{192} & 121 & 14.0 & 184 & 125 & 9.53 \\
 & 432 & 122 & 24.5 & 368 & 104 & 21.2 & 360 & 113 & 13.3 \\
\midrule
\poseidon{}\textsuperscript{a} & \textbf{164} & 60.0 & 7.21 & 298 & 92.0 & \textbf{10.6} & \textbf{180} & 63.0 & 6.60 \\
\poseidontwo{} & \textbf{164} & \textbf{49.8} & \textbf{7.11} & 298 & \textbf{82.8} & \textbf{10.6} & \textbf{180} & \textbf{59.7} & \textbf{6.29} \\
\monolith{} & \multicolumn{3}{c}{--} & \multicolumn{3}{c}{--} & 3536 & 1170 & 14.5 \\
\bottomrule
\end{tabular}%
}

\bigskip

\textit{KoalaBear/BabyBear: $t=24$; Goldilocks: $t=12$}\par\smallskip
\resizebox{\linewidth}{!}{%
\begin{tabular}{lrrr rrr rrr}
\toprule
\multirow{3}{*}{Construction} & \multicolumn{3}{c}{KoalaBear} & \multicolumn{3}{c}{BabyBear} & \multicolumn{3}{c}{Goldilocks} \\
\cmidrule(lr){2-4}\cmidrule(lr){5-7}\cmidrule(lr){8-10}
 & \multirow{2}{*}{Cells} & Prover & Verifier & \multirow{2}{*}{Cells} & Prover & Verifier & \multirow{2}{*}{Cells} & Prover & Verifier \\
 &  & ($\mu$s/perm.) & (ms/proof) &  & ($\mu$s/perm.) & (ms/proof) &  & ($\mu$s/perm.) & (ms/proof) \\
\midrule
\multirow{2}{*}{\design{}} & 336 & 205 & 23.8 & \textbf{288} & 179 & 20.7 & 276 & 211 & 13.2 \\
 & 648 & 183 & 37.3 & 552 & 157 & 32.1 & 540 & 182 & 20.2 \\
\midrule
\poseidon{}\textsuperscript{a} & \textbf{239} & 89.4 & \textbf{9.04} & 450 & 143 & \textbf{14.5} & \textbf{248} & 88.5 & \textbf{7.57} \\
\poseidontwo{} & \textbf{239} & \textbf{70.8} & 9.09 & 450 & \textbf{124} & \textbf{14.5} & \textbf{248} & \textbf{83.9} & 7.64 \\
\monolith{} & \multicolumn{3}{c}{--} & \multicolumn{3}{c}{--} & 3564 & 1190 & 14.4 \\
\bottomrule
\end{tabular}%
}
\par\smallskip
{\footnotesize\raggedright
  \textsuperscript{a} For \poseidon{}, we use the round numbers suggested in the \poseidontwo{} paper~\cite{DBLP:conf/africacrypt/GrassiKS23}. The original \poseidon{} instances were not designed for these small prime fields.
\par}
\end{minipage}
\end{table}

\newpage
\bibliographystyle{splncs04}
\bibliography{./bibliography/references,./bibliography/extra}


\appendix
\renewcommand{\thesection}{A\arabic{section}}
\crefalias{section}{appendix}
\crefalias{subsection}{appendix}
\renewcommand{\theHsection}{appendix.\arabic{section}}
\renewcommand{\theHsubsection}{appendix.\arabic{section}.\arabic{subsection}}
\section*{Appendix}
\section{Proofs for Statistical Analysis}
\label{appendix:proofs-statistical}

This section collects the proofs of the results stated in
\cref{subsection:statisticalattacks}.

\begin{proof}[\cref{lemma:benes_block_diff_prob}]
Consider the system of equations defining the output differences, i.e.,
\begin{align}
    f_0(x_0 + \Delta x_0) - f_0(x_0) + f_2(x_1 + \Delta x_1) - f_2(x_1) &= \Delta y_0 \label{eq:diff1},\\
    f_1(x_0 + \Delta x_0) - f_1(x_0) + f_3(x_1 + \Delta x_1) - f_3(x_1) &= \Delta y_1 \label{eq:diff2}.
\end{align}

Assume \(\Delta x_0 \neq 0\) (the case \(\Delta x_0 = 0, \Delta x_1 \neq 0\) follows analogously by using \cref{eq:diff1} as a polynomial in \(x_1\) with \(f_2\) in place of \(f_0\)). The polynomial expression \(g(x_0) = f_0(x_0 + \Delta x_0) - f_0(x_0)\) evaluates to a polynomial of degree \(d-1\). For any arbitrarily chosen, fixed value of \(x_1 \in \F_p\), \cref{eq:diff1} simplifies to \(g(x_0) = c \in \F_p\). Because \(\deg(g) \leq d-1\), this equation has at most \(d-1\) roots in the field. Since there are \(p\) possible choices for \(x_1\), there are at most \(p \cdot (d-1)\) valid pairs \((x_0, x_1)\) out of \(p^2\) total possible inputs. Thus, the maximum differential probability is bounded by \(\frac{p(d-1)}{p^2} = \frac{d-1}{p}\).
\end{proof}

\begin{proof}[\cref{prop:mds_ext_field}]
Let \(B_\mathrm{in}\) and \(B_\mathrm{out}\) denote the number of active \(\F_{p^2}\) words at the input and output of the MDS matrix. By the definition of the MDS property over \(\F_{p^2}\), any nonzero transition guarantees that the branch number is \(l+1\), meaning \(B_\mathrm{in} + B_\mathrm{out} \geq l+1\). 

Since one \(\F_{p^2}\) element precisely corresponds to one \beneswritten{} block, an active \(\F_{p^2}\) word implies an active \beneswritten{} block. Therefore, the number of active blocks over two rounds is directly given by the branch number \(l+1\). Since the differential probability of the \beneswritten{} block over \(\F_p^2\) remains \(\frac{d-1}{p}\) regardless of the field representation, the maximum \(2\)-round differential probability is strictly bounded by \(\left(\frac{d-1}{p}\right)^{l+1}\).
\end{proof}

\section{Proofs for Algebraic Analysis}
\label{appendix:proofs-algebraic}

This section collects the proofs of the results stated in
\cref{subsection:algebraicattacks}.

\subsection{Quotient Ring Dimension for \design{}}
\label{appendix:ideal-degree-analysis}
\label{subsubsec:gb-analysis}
Let's consider the \design{} permutation both in Sponge and feedforward mode. Let's consider \(n = 2l\) branches and 
let's focus on the $k$-out-of-n instantiation, that is when the output is composed of $k$ elements 
(in feedforward constructions it translates to cropping the output to $k$ elements). 

Let \(c\) denote the number of field elements in the capacity. When the primitive is used in Sponge mode \(c > 0\), we must set other \(n-k-c\) inputs to random values in order to obtain a determined polynomial system. 
When \design\ is used in feedforward mode (e.g. \trunc), there is no capacity (\(c = 0\)) and we must set \(n-k\) inputs to random values in order to obtain a determined polynomial system. Instead of setting those inputs to random values, one could also add additional $n-k-c$ linear relations, either random or chosen by exploiting the structure of the linear layers (e.g.~as happened in round-skipping attacks~\cite{DBLP:journals/iacr/MerzG26}). 
As a result, in both cases, we obtain \(n-k\) fixed inputs and \(k\) fixed outputs. Once this is done, we can proceed with solving the polynomial
system arising from the following CICO-($n$-$k$,$k$) problem~\cite[Algorithm 4]{report-stark-friendly}.  

\begin{definition}[CICO-($n$-$k$,$k$) Problem]
    \label{def:cico-problem-nk-k}
    Let $F: \F_p^n \rightarrow \F_p^n$ be a function and let $k < n$ be an integer. 
    The CICO-($n$-$k$,$k$) problem consists of finding $x \in \F_p^{k} ,y \in \F_p^{n-k}$ such that
    $F(x || \vars{0}^{n-k}) = (y || \vars{0}^{k})$.
\end{definition}

\paragraph{Polynomial Model for \design.}
Starting from the naive polynomial model for \design{}, we are going to modify it in order to remove $n-k$ variables per round. This approach is similar to the technique used in~\cite{arion} for the analysis of Arion.

Let's consider the initial polynomial model for \design{}. Let $R$ be the number of rounds, for $0 \leq i \leq R-1$ we denote with $\vars{x}^{(i)} = \left( x_0^{(i)}, \dots, x_{n-1}^{(i)} \right)$ the inputs to the \((i+1)\)-th round where each round, by specification, is composed of a nonlinear layer \(\mathcal{F}\) and a linear layer \(\llayer\). Let's denote as $\vars{o}$ the output values of the instantiation. Given a valid hash $\vars{b} \in \F_p^k$ produced by \design, we replace \(\vars{o}\) with \(\vars{b}\). The outputs of the initial linear layer (\(M \circ \mathcal{L}(\vars{x})\) as defined in \Cref{defn:design}) are exactly \(\vars{x}^{(0)}\), which are also the inputs to the first round. The original inputs are denoted as \(\vars{a}\).

To build our polynomial modelling we assign weights to the variables as follows: for every $0 \le i \le R-1$, we assign weight \(d^{i}\) to the variables \(\vars{x}^{(i)}\), where $d$ is the maximum degree of the nonlinear layer functions.

\noindent\textbf{The initial affine layer.} Remember that, due to the CICO problem, $n-k$ inputs are fixed. Using \(\overline{\mathcal{L}} = M \circ \mathcal{L}(\cdot)\), \(\overline{\mathcal{L}}^{-1}(\vars{x}^{(0)}) - \vars{a} = 0\). Of those equations, consider only the ones involving \(\vars{a}_{k,n-1}\) which are the $n-k$ fixed inputs. As a result, by Gaussian elimination, we can express \(\vars{x}^{(0)}_{0:n-k-1}\) as a linear combination of the \(k\) unknowns \(\vars{x}^{(0)}_{n-k:n-1}\).  After the initial affine layer, our polynomial model contains 
        \[
            s_{j}^{(0)} := x_j^{(0)} - \ell\left(\vars{x}^{(0)}_{n-k:n-1}\right) \ \text{for} \ 0 \le j \le n-k-1.
        \]
        \noindent\textbf{The round polynomial modelling.}
         Consider the affine layer application $\llayer$. Without losing generality, we use $\llayer(\cdot)$ to denote the pairwise interpretation of the nonlinear layer outputs as elements in \(\F_{p^2}\), the application of the MDS $\llayer$ and the subsequent projection into elements in \(\F_p\), as described in \Cref{defn:round-function}. As it was also explained in \Cref{section:designrationale,section:design}, \(\llayer\) independently mixes the even and odd coordinates of the input. Because the technique that follows is not affected by this behaviour, w.l.o.g. we can assume that all coordinates are mixed together. Indeed, the same methodology can, in a completely analogous way, be applied to any SPN-like design. 

        We model each round by integrating the nonlinear layer \(\nllayer\) and the subsequent linear layer \(\llayer\) into a single polynomial system: 
        \[ \widetilde{\mathcal{F}}^{(i)} = \mathcal{F}^{(i)}(\vars{x}^{(i-1)}) - \llayer^{-1}\left(\vars{x}^{(i)}\right) \ \ \text{ for } 1 \le i \le R. \]
        Consider a generic round $i$. From the previous round ($i-1$), we obtained $n-k$ equations $\vars{s}_{1:n-k}^{(i-1)}$ of the form $s_{j}^{(i-1)} := x_j^{(i-1)} - g_j^{(i-1)}\left(\vars{x}^{(i-1)}_{n-k:n-1}, \vars{x}^{(i-2)}_{n-k:n-1}\right)$ for $0 \le j \le n-k-1$, where $g_j^{(i-1)}$ here denotes generic (even nonlinear) polynomials. When $i = 1$, the equations $\vars{s}_{1:n-k}^{(0)}$ are the ones obtained from the initial affine layer. 
        Those equations can be used to remove the $n-k$ variables $\vars{x}^{(i-1)}_{0:n-k-1}$ from the polynomial system modelling the $i$-th round.
        Note that the leading monomials of equations \(\widetilde{\mathcal{F}}^{(i)}_{n-k:n-1}\), namely \((x_{n-k}^{(i-1)})^d, \dots, (x_{n-1}^{(i-1)})^d\), are not affected by this variable removal step. 

        Equations \(\widetilde{\mathcal{F}}^{(i)}_{0:n-k-1}\) are modified by replacing the variables \(\vars{x}^{(i-1)}_{0:n-k-1}\) with the corresponding polynomials \(g_j^{(i-1)}\left(\vars{x}^{(i-1)}_{n-k:n-1}, \vars{x}^{(i-2)}_{n-k:n-1}\right)\). Since the variables \(\vars{x}^{(i-1)}_{n-k:n-1}\) are involved in \(g_j^{(i-1)}\) with degree 1, terms like \((x_{n-k}^{(i-1)})^d, \dots, (x_{n-1}^{(i-1)})^d\) are generated. However, these are exactly the leading monomials of equations \(\widetilde{\mathcal{F}}^{(i)}_{n-k:n-1}\) and therefore those terms can be replaced by \(\widetilde{\mathcal{F}}^{(i)}_{j} - \lm\left( \widetilde{\mathcal{F}}^{(i)}_{j} \right)\) for \(n-k \le j \le n-1\).

        Consider the modified equations \(\widetilde{\mathcal{F}}^{(i)}_{0:n-k-1}\) after the previous step. Order the monomials such that the monomials \(\vars{x}^{(i)}\) of degree 1, arising from the application of the inverse affine layer, appear first. Under this ordering, we construct the associated coefficient matrix, whose first \(n\) columns are indexed by the variables \(\vars{x}^{(i)}\). Since the matrix has $n-k$ rows, reducing it to row-echelon form and expressing the resulting rows as polynomials yields a new system of equations of the form
        $s_{j}^{(i)} := x_j^{(i)} - g_j^{(i)}\left(\vars{x}^{(i)}_{n-k:n-1}, \vars{x}^{(i-1)}_{n-k:n-1}\right)$ for $0 \le j \le n-k-1$, where $g_j^{(i)}$ here denotes again generic nonlinear polynomials.
        Equations \(\widetilde{\mathcal{F}}^{(i)}_{n-k:n-1}\) also involve the variables \(\vars{x}^{(i)}_{0:n-k-1}\) with degree 1. Those variables can be replaced with the corresponding polynomials \(g_j^{(i)}\left(\vars{x}^{(i)}_{n-k:n-1}, \vars{x}^{(i-1)}_{n-k:n-1}\right)\) for \(0 \le j \le n-k-1\). 
        The final polynomial system contains the just modified equations \(\widetilde{\mathcal{F}}^{(i)}_{n-k:n-1}\). The equations \(s_{j}^{(i)}\) are used for the variable removal step of the next round.

        Thanks to the variable removal steps and thanks to the variable weights we have chosen before, the equations \(\widetilde{\mathcal{F}}^{(i)}_{n-k:n-1}\) contain the terms \[
         \widetilde{\mathcal{F}}^{(i)}_{j} = \left(x^{(i-1)}_j\right)^d  + \sum_{\substack{d_1+\dots+d_k=d\\(d_1,\dots,d_k)\neq(\vars{0}^{n-k-j},d,\vars{0}^{n-j})}} \prod_{z=n-k}^{n-1} \left(x^{(i-1)}_z\right)^{d_z}
          \ \ \text{ for } n-k \le j \le n-1.
        \] Considering the chosen weights, those terms appear first, meaning that the first monomials of those equations are the monomials of degree \(d\) in the variables \(\vars{x}^{(i-1)}_{n-k:n-1}\). As a result, \(\deg\left(\widetilde{\mathcal{F}}^{(i)}_{j}\right) = d\) for \(n-k \le j \le n-1\).
        
    \noindent\textbf{The last round.} 
    The final round is treated analogously, with the distinction that the transformation is applied in the forward direction ($\llayer$ is applied rather than its inverse) and only the equations corresponding to the $k$ known (and fixed) output values are retained.
        In the case of Sponge mode, those equations are
        \[ \widetilde{\mathcal{F}}^{(r)}_{0:k-1} = \llayer\left(\mathcal{F}^{(r)}\right)_{0:k-1} - \vars{b} \] 
        whilst in feedforward mode they are
        \[ \widetilde{\mathcal{F}}^{(r)}_{0:k-1} = \llayer\left(\mathcal{F}^{(r)}\right)_{0:k-1}  + \vars{a}_{0:k-1} - \vars{b}. \] 
        The variable removal step proceeds as before. Note that in the feedforward case, even the variables \(\vars{x}^{(0)}_{0:k-1}\) must be replaced by using the replacement equations \(s_{j}^{(0)}\) obtained from the initial affine layer. 
        
\begin{simpleframe}{\textbf{The final polynomial system}}
    Let \(R\) denote the number of rounds and $k$ the number of constrained outputs. After the variable removal steps, we obtain a polynomial system with $Rk$ equations of total degree $d$ ($k$ equations per round) and $Rk$ variables.
\end{simpleframe}

\subsection{Bounding the Quotient Ring Dimension}

Before proceeding with the quotient ring dimension analysis, we need to recall some theorems and definitions related to B{\'e}zout bounds.

\begin{theorem}[B{\'e}zout's Theorem]
    \label{th}
    Let \(I = \langle f_0, \dots, f_{m-1} \rangle \subset \F_p[x_0, \dots, x_{m-1}]\) be a zero-dimensional ideal and let \(d_i = \deg(f_i)\) denote the total degree of $f_i$, for $i = 0, \dots, m-1$. If the system has a finite number of solutions (following from the finiteness theorem and zero-dimensional ideal), then the number of solutions (counting multiplicities) over the algebraic closure \(\overline{\F_p}\) is at most the product of the degrees of the polynomials, i.e., \(\prod_{i=0}^{m-1} d_i\).
\end{theorem}

The number of solutions is exactly the quotient ring dimension that we want to bound and prove. It can be shown that if the number of solutions in $\overline{\F_p}$ is finite, that is, if $I$ is zero-dimensional, then the homogenization only adds a finite number of additional solutions at infinity over the projective space \(\mathbb{P}^n(\overline{\F_p})\)~\cite{MW83}. Thus, the bound is sharp if and only if the number of solutions at infinity is zero (except the trivial solution at the origin).

\begin{proposition}[Quotient Ring Dimension of the CICO-($n$-$k$,$k$) Problem for \design{}]
    \label{prop:ideal-degree}
    Let \(I\) be the ideal generated by the polynomial system modelling the CICO-($n$-$k$,$k$) problem for \design{}. With high probability, the quotient ring dimension of \(I\) is exactly \(d^{kR}\), where \(R\) is the number of rounds and $k$ is the number of constrained outputs.
\end{proposition}

\begin{proof}
As the leading monomials of the equations in the defined polynomial system have degree \(d\), the B{\'e}zout bound is given as \(d^{kR}\), where $R$ is the number of rounds. We show that the B{\'e}zout bound is tight by showing that the number of solutions at infinity is zero.

Because we are dealing with weighted orders, we need to consider a weighted homogenization procedure. Let's denote with \(w_i\) the weight of variable \(x_i\) and with \(d_i\) the total degree of polynomial \(f_i\). The weighted homogenization of \(f_i\) is defined as follows:
\[\widetilde{f}_i = x_0^{d_i} f_i\left(\frac{x_1}{x_0^{w_1}}, \dots, \frac{x_n}{x_0^{w_n}}\right),\] where \(x_0\) is the homogenization variable. The homogenized ideal is then defined as \(\widetilde{I} = \langle \widetilde{f}_1, \dots, \widetilde{f}_m \rangle\).

Counting the solutions at infinity is equivalent to counting the solutions of the homogenized system \(\widetilde{I}\) that satisfy \(x_0 = 0\). After replacing \(x_0\) with 0, we obtain, per round, $k$ homogeneous equations of the form \[
    \sum_{\substack{d_1+\dots+d_k=d}} \prod_{z=n-k}^{n-1} \left(x^{(i-1)}_z\right)^{d_z} = 0,
\] where coefficients are omitted for simplicity.

In other words, we obtain $R$ homogeneous systems of $k$ equations in $k$ variables, where the degree of each equation is \(d\). 
It is well-known in the literature~\cite{CoxLitOsh97} that, with high probability, the only solution of a dense homogeneous system of $k$ independent (otherwise it would not be a determined polynomial system) equations of degree $d$ in $k$ variables is the trivial solution at the origin. Hence, each of those systems has only the trivial solution at the origin.
As a result, the number of solutions at infinity is zero (as the trivial solution is not considered a solution at $\infty$) and the B{\'e}zout bound is tight.
\end{proof}

In addition, as the ideal is zero-dimensional and has no solutions at infinity, by~\cite[Theorem 8]{DBLP:conf/eurocrypt/CampaR25} and~\cite[Theorem 1.2]{DBLP:journals/moc/CoxD23} the ideal has Shape form, thus justifying a possible application of the FGLM algorithm to solve the system after GB computation. 
The results on the quotient ring dimension were supported by experimental evidence obtained by computing the GB of the polynomial system modelling the CICO-($n$-$k$,$k$) problem for toy versions of \design{}.

\subsection{Interpolation Attack and Proofs}
\label{appendix:proofs-algebraic-degree}

\begin{lemma}[Degree after each \beneswritten{} Block.]
    \label{lem:degree-after-benes-block}
    Consider the \design{} round function as in \cref{defn:round-function}. All the inputs are pair-wise passing through the \beneswritten{} (\cref{defn:design-non-linear-block-design}) nonlinear block. Let \(x,y\) be the inputs to a single \(\benes\) block and let \(d = \deg(f_0) = \deg(f_2) = \deg(f_3)\). Then, after applying one \beneswritten{} block, the degree of the outputs is \(\max(d \cdot \deg(x), d \cdot \deg(y))\) and \(\max(\deg(f_1) \cdot \deg(x), d \cdot \deg(y))\), respectively. 
\end{lemma}

\begin{proof}
Consider the \design{} round function as in \cref{defn:round-function}. More specifically, let us recall the corresponding nonlinear layer: \[\nllayer(x_0, \dots, x_{2l-1}) = \concat_{i=0}^{l-1} \benes(x_{2i}, x_{2i+1}),\] where \(\benes\) is defined as in \cref{defn:design-non-linear-block-design} (\(\benes(x, y) = (u, v) = (f_0(x) + f_2(y), f_1(x) + f_3(y))\)). Let \(d\) be the degree of the nonlinear function. From \cref{defn:design-non-linear-block-design}, we can notice that \(\deg(u) = \max(d \cdot \deg(x), d \cdot \deg(y))\) and \(\deg(v) = \max(\deg(f_1) \cdot \deg(x), d \cdot \deg(y))\). 
\end{proof}

\begin{lemma}[Degree after the MDS Application.]
\label{lem:degree-after-mds}
    Let \(2l\) be the number of branches (inputs) of the \design{} round function and let \(M \in \F_p^{l \times l}\) be the MDS matrix. Let denote as \(x_0, \dots, x_{2l-1}\) the inputs to the MDS matrix (before interpreting them as elements in \(F_{p^2}\)) and \(y_0, \dots, y_{2l-1}\) the corresponding outputs (after interpreting them as elements in \(\F_p\)). 
    After applying the MDS matrix, \(\deg(y_{2i}) = \max(x_0, x_2, \dots, x_{2l-2})\) and \(\deg(y_{2i+1}) = \max(x_1, x_3, \dots, x_{2l-1})\) for \(1 \le i \le l-1\).
\end{lemma}

\begin{proof}
   Let \(M \in \F_p^{l \times l}\) be the MDS matrix and let \(x_0, \dots, x_{2l-1}\) be the inputs to the MDS matrix (before interpreting them as elements in \(F_{p^2}\)) and \(y_0, \dots, y_{2l-1}\) the corresponding outputs (after interpreting them as elements in \(\F_p\)). 
    Recall that the MDS matrix application is applied after interpreting the $\benes$ outputs as elements in \(\F_{p^2}\), namely \(y_{2i}X + y_{2i+1}\) for \(1 \le i \le l-1\), where $X$ denotes the root of the irreducible polynomial defining the extension. 
    As the MDS matrix is over \(\F_p\), this application is equivalent to applying the MDS matrix on the variables \(x_0, x_2, \dots, x_{2l-2}\) and on the variables \(x_1, x_3, \dots, x_{2l-1}\) separately. Hence, after the MDS application, the even indexed outputs are linear combinations of the variables \(y_0, y_2, \dots, y_{2l-2}\), while the odd-indexed variables are a linear combination of the variables \(y_1, y_3, \dots, y_{2l-1}\). Because a linear combination does not alter the degree of the input variables, \(\deg(y_{2i}) = \max(x_0, x_2, \dots, x_{2l-2})\) and \(\deg(y_{2i+1}) = \max(x_1, x_3, \dots, x_{2l-1})\) for \(1 \le i \le l-1\). 
\end{proof}

\begin{proof}[\cref{lem:degree-after-r-rounds}]
Let \(2l\) be the number of branches (inputs) of the \design{} round function and let \(d = \deg(f_0) = \deg(f_2) = \deg(f_3)\) and \(\beta = \deg(f_1)\).
    After the initial linear layer application (as defined in \Cref{defn:design}), each branch is a linear combination of the input variables. Hence, the degree of each branch is \(1\).

    For the sake of simplicity, let denote as \(x_0, \dots, x_{2l-1}\) the inputs to the first \benes\ function (the nonlinear layer), as \(y_0, \dots, y_{2l-1}\) the corresponding outputs and as \(z_0, \dots, z_{2l-1}\) the outputs of the subsequent MDS application. We use \(\deg(\vars{x})\) to denote the maximum degree of the variables \(x_0, \dots, x_{2l-1}\).

    After applying the nonlinear layer, from \cref{lem:degree-after-benes-block} we know that \(\deg(y_{2i}) = \max(d \cdot \deg(x_{2i}), d \cdot \deg(x_{2i+1}))\) and \(\deg(y_{2i+1}) = \max(\deg(f_1) \cdot \deg(x_{2i}), d \cdot \deg(x_{2i+1}))\) for \(1 \le i \le l-1\). But, as \(\deg(x_j) = \deg(x_k)\) for all \(1 \le j < k \le 2l-1\), \(\deg(y_{2i}) = \deg(y_{2i+1}) = d \cdot \deg(\vars{x})\) for all \(1 \le i \le l-1\). As a result, because in the first round \(\deg(\vars{x}) = 1\), \(\deg(y_{2i}) = \deg(y_{2i+1}) = d\). 
    Following \Cref{lem:degree-after-mds}, we know that the degrees of the inputs are not affected, hence, after the first round (comprehending also the MDS application), \(\deg(x_j) = d\) for all \(1 \le j \le 2l-1\).

    Repeating the same reasoning, as the degree of each branch is the same at the beginning of each round, we can conclude that after \(R\) rounds, the degree of each output is given as \(d^{R}\).
\end{proof}

\section{MDS Matrix Implementation}
\label{appendix:mds-implementation}

\paragraph{Circulants as Polynomial Products.}
For the algebraic explanation, it is convenient to use a first-column convention. A circulant matrix is determined by a sequence \(c = (c_0, c_1, \ldots, c_{n-1}) \in \F_p^n\); set
\[
    M_{i,j} \;=\; c_{(i - j) \bmod n}.
\]
The implementation stores first rows instead. This is the transpose convention, equivalently replacing \(X\) by \(X^{-1}\) in the polynomial description; transposition preserves the MDS property and does not change the adder cost.

Identify the state space \(\F_p^n\) with
\[
    R_n \;:=\; \F_p[X] / (X^n - 1)
\]
through the linear map \(\varphi\) sending \((a_0,\ldots,a_{n-1})\) to \(\sum_i a_i X^i\). If \(f(X)\) represents the circulant column \(c\) and \(g(X)\) represents the input vector, then
\[
    \varphi(M \cdot x) = f(X) \cdot g(X) \pmod{X^n - 1}.
\]
Thus a circulant matrix--vector product is exactly a cyclic convolution.

\paragraph{CRT Decomposition.}
Over \(\Z\), and over \(\F_p\) for primes \(p \nmid n\), the polynomial \(X^n - 1\) factorises into distinct cyclotomic factors,
\[
    X^n - 1 \;=\; \prod_{d \mid n} \Phi_d(X).
\]
Because these factors are pairwise coprime, the Chinese Remainder Theorem gives
\[
    R_n \;\xrightarrow{\;\sim\;}\; \bigoplus_{d \mid n} \F_p[X] / \Phi_d(X), \qquad
    f \;\longmapsto\; (f \bmod \Phi_d)_{d \mid n}.
\]
Multiplication in \(R_n\) is therefore reducible to products in lower-degree quotient rings. In a field NTT this would be done by evaluating at roots of unity in \(\F_p\). For the hardware path here, the length-\(4\) stage is a real FFT: the \(u=1\) and \(u=-1\) bins use only signs, and the conjugate \(u^2=-1\) bins are kept as two coordinates representing elements of \(\F_p[U]/(U^2+1)\). This keeps the transform itself in shifts, sign changes, and additions.

The decomposition is classical. In FFT terminology it is related to the Good--Thomas, or Prime-Factor, Algorithm; in cyclic-convolution terminology it is the CRT-based multidimensional convolution construction of Agarwal and Cooley~\cite{AgarwalCooley1977}. Bernstein describes the same ring-homomorphism viewpoint as Good's trick~\cite{BernsteinM3}, and algebraic signal processing systematizes this style of derivation as polynomial algebras decomposed by the CRT~\cite{PuschelMoura2008}. In the AO-hash setting, this integer-FFT MDS technique was introduced by Jacqueline Nabaglo in Plonky2, and later appeared in RPO and TIP5~\cite{plonky2,DBLP:journals/iacr/AshurKMST22,DBLP:journals/iacr/SzepieniecLST23}.

\paragraph{The \(n = 12\) Schedule.}
For \(n = 12\), set \(U=X^3\). Since \(U^4=1\), every state polynomial can be written as
\[
    g(X) =
    g_0(U) + Xg_1(U) + X^2 g_2(U),
    \qquad
    g_b(U) = \sum_{a=0}^{3} g_{3a+b} U^a .
\]
This is exactly the stride-\(3\) layout used by the implementation:
\[
    [s_0, s_3, s_6, s_9],\quad [s_1, s_4, s_7, s_{10}],\quad [s_2, s_5, s_8, s_{11}],
\]
with three length-\(4\) real FFTs in total, one for each polynomial \(g_b(U)\).

After the FFT, fix a bin \(U=u\). The residual variable \(Y\), the image of \(X\) in that bin, satisfies \(Y^3=u\). Multiplication by the circulant row is therefore a length-\(3\) product in \(\F_p[Y]/(Y^3-u)\):
\begin{itemize}
    \item \emph{DC bin (\(u=1\)).} The quotient is \(\F_p[Y]/(Y^3-1)\), so the block is a cyclic length-\(3\) convolution. This is \texttt{block1\_3}.
    \item \emph{Nyquist bin (\(u=-1\)).} The quotient is \(\F_p[Y]/(Y^3+1)\), so the block is a negacyclic length-\(3\) convolution. This is \texttt{block3\_3}.
    \item \emph{Complex bin (\(u^2=-1\)).} The quotient is represented over \(\F_p[U]/(U^2+1)\) with the relation \(Y^3=u\). In the implementation's FFT-\(4\) sign convention, this twisted length-\(3\) product is \texttt{block2\_3}.
\end{itemize}
The inverse length-\(4\) real FFT interpolates the three transformed \(g_b(U)\) polynomials back to coefficients in \(U\), producing the output in the original state order.

\paragraph{The \(n = 8\) Schedule.}
For \(n = 8\), set \(U=X^2\). Then \(U^4=1\) and
\[
    g(X) =
    g_0(U) + Xg_1(U),
    \qquad
    g_b(U) = \sum_{a=0}^{3} g_{2a+b} U^a .
\]
The implementation therefore splits the input into even and odd positions and applies one real FFT-\(4\) to each group, for two length-\(4\) real FFTs in total. For each FFT-\(4\) frequency type, the even-bin value and the odd-bin value form one length-\(2\) vector. Fixing \(U=u\), the remaining variable \(Y\), again the image of \(X\), satisfies \(Y^2=u\):
\begin{itemize}
    \item \emph{DC bin (\(u=1\)).} The quotient is \(\F_p[Y]/(Y^2-1)\), implemented by \texttt{block1\_2}.
    \item \emph{Nyquist bin (\(u=-1\)).} The quotient is \(\F_p[Y]/(Y^2+1)\), implemented by \texttt{block3\_2}.
    \item \emph{Complex bin (\(u^2=-1\)).} The quotient is represented over \(\F_p[U]/(U^2+1)\) with \(Y^2=u\), implemented by \texttt{block2\_2}.
\end{itemize}
The output is recovered with two inverse length-\(4\) real FFTs. As in the \(n=12\) case, the FFT/IFFT structure is fixed by the dimension; only the small block constants are matrix-specific.

\paragraph{FFT-Block Search Parameterisation.}
The selection procedure in \cref{section:designrationale} samples candidates for \(n \in \{8,12\}\) directly in the frequency domain rather than in time. Each candidate is specified by the small integer constants of the three FFT-\(4\) frequency blocks --- the \(u=1\) (cyclic), \(u^2=-1\) (twisted, over \(\F_p[U]/(U^2+1)\)), and \(u=-1\) (negacyclic) blocks described above --- of length \(3\) for \(n=12\) and length \(2\) for \(n=8\). The candidate time-domain row is then reconstructed by applying the inverse FFT-\(4\) along the length-\(4\) axis and converting the resulting first column to the implementation's first-row convention. This is the parameterisation referred to in the selection procedure.

\section{Additional Software Performance Results}
\label{appendix:performance}

We implement \design{} in the \enquote{ZK-friendly Hash Zoo}, a benchmark framework for ZK-friendly hash functions used in previous literature~\cite{DBLP:journals/iacr/BarbaraGKLRSW21,DBLP:journals/tosc/GrassiKLRSW24,DBLP:journals/tches/BouvierGKKRSS25}.\footnote{\url{https://extgit.isec.tugraz.at/krypto/zkfriendlyhashzoo/}} It provides many hash functions with different underlying primes, although for small primes only Goldilocks is supported across a variety of them.

\begin{table}[ht]
  \centering
  \caption{Software performance of compressing $\approx 512$ bits, as implemented in the ZK-friendly Hash Zoo in the Goldilocks prime field. The reported results are in nanoseconds and averaged over 10000 runs.
  State sizes are $t=8$ and $t=12$ as typical. Tip5 is only supported with $t=16$.}\label{tab:comparison_software_zk_zoo}
  \begin{multicols}{2}
      \begin{tabular}{c|cc}
        \toprule
        Hash            & \multicolumn{2}{c}{Runtime (ns)} \\\cmidrule(lr){2-3}
                        & (t = 8)     & (t = 12)           \\\midrule
        \rescueprime{} \cite{DBLP:journals/tosc/AlyABDS20,DBLP:journals/iacr/SzepieniecAD20}  & 12\,276     & 19\,231            \\
        Neptune         & 2206        & 3357               \\
        Griffin         & 1818        & 1996               \\
        GMiMC           & 2371        & 4993               \\
        \poseidon{}     & 1902        & 3257 \\
        \bottomrule
      \end{tabular}
    \columnbreak

    \begin{tabular}{c|cc}
      \toprule
      Hash            & \multicolumn{2}{c}{Runtime (ns)} \\\cmidrule(lr){2-3}
                      & (t = 8)     & (t = 12)           \\\midrule
      \cellcolor{gray!25}\design{}       & \cellcolor{gray!25}1199         & \cellcolor{gray!25}1551                \\
      \poseidontwo{}       & 949         & 1266               \\
      Tip5            & \multicolumn{2}{c}{466*}         \\
      Tip4'           & -           & 259                \\
      \monolith{}        & 131         & 247                \\
      \bottomrule
    \end{tabular}
    \end{multicols}
\end{table}

\cref{tab:comparison_software_zk_zoo} shows that \design{} is slower than \poseidontwo{} at both state sizes: it takes $1199$ ns versus $949$ ns for $t=8$, and $1551$ ns versus $1266$ ns for $t=12$, corresponding to approximately $26\%$ and $23\%$ higher runtimes, respectively. Among the designs using only low-degree nonlinear functions, \poseidontwo{} is the fastest in this comparison, while \design{} remains faster than \poseidon{} and GMiMC. For the \poseidon{} implementation, the authors use Cauchy MDS matrices, where $M_{i,j}=\frac{1}{x_i+y_j}$. However, the library does not use any specific optimizations but rather evaluates it as an unstructured dense matrix. For the internal layers, the framework does use the efficient representation proposed by the \poseidon{} paper~\cite[Appendix~B]{DBLP:conf/uss/0001KR0S21}.

The other design strategies exhibit familiar trade-offs. Hash functions relying on the power maps $x^{d}$ and $x^{1/d}$ incur a substantial plain evaluation cost for the inverse power, as seen for \rescueprime{} and Griffin. Hash functions based on the split-and-lookup paradigm, with nonlinear layers built from non-algebraic bit transformations, achieve the lowest plain evaluation times in this comparison, but their efficient arithmetizations rely on lookup support that may not be available in all proof systems.

\end{document}